\documentclass[11pt,a4paper]{article}
\usepackage[margin=1in]{geometry}
\usepackage{microtype}
\usepackage{graphicx}
\usepackage{subcaption}
\usepackage{booktabs} 

\usepackage{hyperref}
\usepackage{algorithm}
\usepackage{algorithmic}
\usepackage{xcolor}
\usepackage{amsmath}
\usepackage{amssymb}
\usepackage{mathtools}
\usepackage{amsthm}
\usepackage{mleftright}
\usepackage{varwidth}
\usepackage[capitalize,noabbrev]{cleveref}

\usepackage{thmtools}
\usepackage{thm-restate}
\newtheorem{theorem}{Theorem}[section]

\newtheorem{lemma}[theorem]{Lemma}
\newtheorem{corollary}[theorem]{Corollary}
\theoremstyle{definition}
\newtheorem{definition}[theorem]{Definition}

\theoremstyle{remark}

\usepackage{authblk}

\title{Active Regression for Single-Index Models with Unknown Link Functions}

\author[1]{Chansophea Wathanak In\thanks{Supported by an NTU Research Scholarship.}}
\author[1]{Yi Li\thanks{Supported in part by the Singapore Ministry of Education AcRF Tier 1 grant RG21/25.}}
\author[2]{Wai Ming Tai}
\author[3]{Xuan Wu\thanks{Supported by the Singapore Ministry of Education AcRF Tier 1 grant RG21/25 while at Nanyang Technological University, where part of this work was completed.}}

\affil[1]{School of Physical and Mathematical Sciences, Nanyang Technological University, Singapore\\
\texttt{inch0002@e.ntu.edu.sg, yili@ntu.edu.sg}}
\affil[2]{Independent Researcher\\
\texttt{taiwaiming2003@gmail.com}}
\affil[3]{John Hopcroft Center for Computer Science, Shanghai Jiao Tong University, China\\
\texttt{wuxuan2026@sjtu.edu.cn}}

\usepackage{enumitem}

\newcommand{\eps}{\epsilon}

\newcommand{\eat}[1]{}

\DeclareMathOperator*{\argmin}{arg\,min}
\DeclareMathOperator{\sgn}{sgn}

\newcommand{\R}{\mathbb{R}}
\newcommand{\bS}{\mathbb{S}}
\DeclareMathOperator*{\E}{\mathbb{E}}

\DeclareMathOperator{\diam}{diam}

\newcommand{\cD}{\mathcal{D}}

\newcommand{\cN}{\mathcal{N}}

\newcommand{\OPT}{\ensuremath{\mathsf{OPT}}}

\newcommand{\lip}{\ensuremath{\mathsf{Lip}}}

\DeclareMathOperator{\poly}{poly}

\DeclareMathOperator{\diag}{diag}
\DeclareMathOperator{\Ber}{Ber}
\DeclareMathOperator{\Bin}{Bin}

\DeclarePairedDelimiter{\abs}{\lvert}{\rvert}

\DeclarePairedDelimiter{\norm}{\lVert}{\rVert}

\DeclarePairedDelimiterX{\inner}[2]{\langle}{\rangle}{#1 , #2}
\DeclarePairedDelimiterXPP\normp[1]{}{\lVert}{\rVert}{_p}{#1}

\newcounter{cases}
\newcounter{subcases}[cases]
\newenvironment{mycase}
{
	\setcounter{cases}{0}
	\setcounter{subcases}{0}
	\newcommand{\case}
	{
		\par\indent\stepcounter{cases}\textbf{Case \thecases.}
	}
	
}
{
	\par
}
\renewcommand*\thecases{\arabic{cases}}

\begin{document}

\date{}
\maketitle

\begin{abstract}
    This paper studies active regression for single-index models under general $\ell_p$-loss with an unknown $1$-Lipschitz link function $f$, formulated as $\min_{f,x} \|f(Ax)-b\|_p^p$ with full access to $A$ but coordinate-query access to $b$. 
    Prior work established upper bounds for known link functions for all $p\geq 1$ and for unknown link functions only in the $p=2$ case, together with lower bounds for $p\leq 2$.
    This work addresses the more challenging setting of unknown link functions and general $p \geq 1$.
    A non-adaptive sampling algorithm is presented that achieves a 
    $(1+\eps)$-approximation using $O(d^{p/2\vee 1}/\eps^{p\vee 2}\poly\log(n/\eps))$ queries. 
    Nearly tight lower bounds are also established for $p>2$. These results close much of the remaining gap in active $\ell_p$-regression for single-index models.
\end{abstract}

\section{Introduction}
The $\ell_p$-regression problem is a fundamental task in randomized numerical linear algebra (RandNLA). Given a data matrix $A\in\R^{n\times d}$ with $n\gg d$ and a label vector $b\in \R^n$, the goal is to find a vector $x$ that minimizes the residual $\norm{Ax-b}_p$. A standard paradigm to reduce the computational complexity is row sampling: one constructs a (randomized) row-sampling matrix $S$, meaning that each row has exactly one nonzero entry, such that the solution $\hat{x}$ to the sketched problem, i.e.\ $\hat{x} = \argmin_x \norm{S(Ax-b)}_p$, approximately solves the original problem in the sense that 
\[
\norm{A\hat{x}-b}_p \leq (1+\eps)\min_x \norm{Ax - b}_p
\]
with probability at least $0.9$. 
After over a decade of intensive research, the sampling complexity for $\ell_p$-regression is now well-understood for all values of $p\geq 1$; see, e.g.,~\cite{w14,taisuke:thesis}. In particular, $S$ can be constructed by sampling rows of $(A\ \ b)$ according to the Lewis weights (a generalization of leverage scores)~\cite{LT91,Cohen15}. The resulting sample complexity is $\tilde{O}(d/\eps^2)$ rows for $p=1$ and $\tilde{O}(d^{\max\{p/2,1\}}/\eps)$ rows for $p > 1$~\cite{taisuke:thesis}. 

Remarkably, this framework extends to \emph{active regression}, where the matrix $A$ is fully known but access to $b$ is restricted to coordinate queries. In this setting, one can form $S$ using only Lewis weights of $A$. The sample complexity to achieve the same $(1+\eps)$-approximation guarantee is $\tilde{O}(d/\eps^2)$ queries for $p=1$~\cite{price:regression} and $\tilde{O}(d^{\max\{p/2,1\}}/\eps^{p-1})$ queries for $p > 1$~\cite{musco22}. Both bounds are known to be tight up to logarithmic factors, even when the samples are allowed to be adaptive~\cite{price:regression,taisuke:thesis}.

Despite the mature understanding of linear $\ell_p$-regressions, linear models cannot capture common nonlinear behaviours, such as piecewise-linear activations (e.g.\ ReLU) arising from neural architectures. 
This limitation has motivated a shift to nonlinear regression problems in recent years, notably the \emph{single-index model}, which generalizes linear regression by composing the linear predictor $Ax$ with a link function $f$. 

In this paper, we consider Lipschitz link functions and study the single-index problem as a purely deterministic regression task:
\begin{equation}\label{eqn:problem_formulation}
    \min_{f\in\lip_1, x\in \R^d} \norm{f(Ax) - b}_p,
\end{equation}
where 
\[
    \lip_L = \{f\in C(\R): |f(x) - f(y)|\leq L|x-y| \text{ for all $x,y\in\R$ and }f(0) = 0\}
\]
and $f$ is applied entrywise, i.e., $f(z) = (f(z_1),\dots,f(z_n))$ for $z\in\R^n$. The goal is to approximately solve the regression problem in the active setting introduced earlier using as few queries to coordinate of $b$ as possible.

Gajjar, Hegde, and Musco~\cite{gajjar23a} first investigated the single-index regression problem from a RandNLA perspective in the case of $p=2$ with a \emph{known} link function $f$, i.e., the minimization in \eqref{eqn:problem_formulation} is only over $x\in\R^d$. With a query complexity of $O(d^2/\eps^4)$, their algorithm outputs $\hat{x}$ that satisfies a constant-factor error bound
\[
    \norm{f(A\hat{x})-b}_p^p \leq C\left(\norm{f(Ax^\ast) - b}_p^p + \eps\norm{Ax^\ast}_p^p\right),
\]
with probability at least $0.9$. Here, $C$ is an absolute constant and $x^\ast$ denotes an optimal solution. (If the optimal solution is not unique, we choose $x^\ast$ among the optimal solutions to minimize $\norm{Ax^\ast}_p^p$.) It was also shown in \cite{gajjar23a} that the additive error is necessary if one aims for query complexity polynomial in $d$.
This complexity was subsequently improved to $\tilde{O}(d/\eps^2)$ in~\cite{colt} and generalized to $\tilde{O}(d^{\max\{p/2,1\}}/\eps^4)$ for $p\geq 1$ in~\cite{ICLR24}. The current state-of-the-art result in~\cite{iclr2025} achieves a query complexity of $\tilde{O}(d^{\max\{p/2,1\}}/\eps^{\max\{p,2\}})$ for $p\geq 1$ and a $(1+\eps)$-approximation, i.e.\ the output $\hat{x}$ satisfies
\begin{equation}\label{eqn:known_f_obj}
    \norm{f(A\hat{x})-b}_p^p \leq (1+\eps)\norm{f(Ax^\ast) - b}_p^p + \eps\norm{Ax^\ast}_p^p
\end{equation}
with probability at least $0.9$. Moreover, Li and Tai~\cite{iclr2025} showed that this query complexity is tight up to logarithmic factors for $1 \leq p\leq 2$ and that the $1/\eps^p$ dependence is tight for $p > 2$, leaving open the question of a matching lower bound for $p > 2$.

Much less is known when the link function is \emph{unknown}, which is exactly \eqref{eqn:problem_formulation} and our focus in this paper. To date, the only result in this setting, from~\cite{colt}, gives a query complexity of $O(d/\eps^2\poly\log n)$ and a constant-factor approximation for the special case $p=2$ only. The algorithm there outputs a pair $(\hat{f},\hat{x})$ satisfying the guarantee:
\[
    \norm{\hat{f}(A\hat{x})-b}_p^p \leq C\left(\norm{f^\ast(Ax^\ast) - b}_p^p + \eps\norm{Ax^\ast}_p^p\right),
\]
where $C$ is an absolute constant and $(f^\ast, x^\ast)$ denote the optimal solution to \eqref{eqn:problem_formulation}.
However, their analysis is very specific for $p=2$ and does not extend to other values of $p$ easily. Consequently, a significant gap remains in the understanding of single-index regression with unknown link functions.

\subsection{Our Results}
Our first contribution is an upper bound for the active single-index regression problem \eqref{eqn:problem_formulation}, achieving a $(1+\eps)$-approximation for general $p \geq 1$.
This result strictly improves the result in~\cite{colt} in both the approximation guarantee and the range of $p$.
In particular, we show that when the link function $f$ is unknown, the query complexity 
is only $\log n$ factors larger than the case of known $f$, consistent with known behaviour for the special case $p=2$. 

\begin{theorem}[Informal version of \Cref{thm:1+eps approx}]
  Let $A \in \R^{n \times d}$, $b \in \R^n$ and $\eps > 0$.   
  There exists a randomized algorithm which, with probability at least $0.9$, makes $O(d^{1 \vee \frac{p}{2}}/\eps^{p \vee 2}\cdot \poly\ln n)$ non-adaptive queries to the entries of $b$ and returns $(\hat{f},\hat{x}) \in \lip_1 \times \mathbb{R}^d$ satisfying the ($1+\eps$)-approximation guarantee 
  \begin{equation}\label{eqn:main_obj}
    \norm{\hat{f}(A\hat{x})-b}_p^p \leq (1+\eps)\norm{f^\ast(Ax^\ast) - b}_p^p + \eps\norm{Ax^\ast}_p^p.
\end{equation}
\end{theorem}

We remark that our algorithm remains valid if, in the problem formulation~\eqref{eqn:problem_formulation}, the link function $f$ is restricted to a subset $\mathcal{F} \subset \lip_1$ instead of ranging over the entire $\lip_1$.

Our second contribution gives a lower bound for the case of a known link function when $p>2$. In the large-dimensional regime, the known upper bound $\widetilde O(d^{p/2}/\eps^p)$ is tight up to logarithmic factors, even when the queries may be adaptive.
\begin{theorem}[Informal version of \Cref{thm:large-d-code-lb}]
	Let $p>2$ be constant. Suppose that $\eps>0$ is sufficiently small and $d\gtrsim_p \log(d/\eps)$, there exist a deterministic link function $f\in\lip_1$ and a deterministic matrix $A$ such that every randomized algorithm that queries coordinates of $b$ and outputs $\hat{x}\in\R^d$ which satisfies (\ref{eqn:known_f_obj}) with probability at least $4/5$ must query	
	\[
		\Omega_p\mleft(\frac{d^{p/2}}{\eps^p(\log(d/\eps))^{p/2}}\mright)
	\]
	entries of $b$.
\end{theorem}
We remark that when $d\lesssim_p \log(d/\eps)$, the existing lower bound $\Omega(1/\eps^p)$ from \cite{iclr2025} already matches the upper bound $\tilde{O}(d^{p/2}/\eps^p)$ up to logarithmic factors. Thus, the regime of interest is $d\gtrsim_p \log(d/\eps)$, as assumed in the preceding theorem.

\subsection{Technique Overview}\label{sec:techniques}

\paragraph{Upper Bound.} The algorithm is a simple, non-adaptive sampling scheme following the template in~\cite{iclr2025,colt}: 
We perform weighted sampling of the coordinates of $b$ where the weights are determined by the Lewis weights of $A$.
While the algorithm is conceptually simple, the main technical contribution lies in the analysis required to handle the unknown link function.

Our analysis builds on the framework introduced in~\cite{iclr2025}, incorporating the ideas for unknown link functions from~\cite{colt}. To simplify the exposition, we assume that $A$ has uniformly bounded Lewis weights $O(d/n)$ and let $\psi_1,\dots,\psi_n$ be i.i.d.\ Bernoulli variables such that $\E \psi_i = \alpha$. Let $T$ be a bounded region in $\lip_1\times \R^d$ and $v$ be a fixed vector in $\R^n$.
The key technical task is to control a uniform sampling error of the following form:
\begin{equation}\label{eqn:Psi_intro}
  \Psi = \sup_{(f,x)\in T} \abs*{ \sum_{i = 1}^n 
    \left(\frac{1}{\alpha}\psi_i  - 1\right) Z_i(f,x) },
\end{equation}
where 
\[
    Z_i(f,x) = \abs{(f(Ax)-v)_i}^p - \abs{(f^*(A x^*)-v)_i}^p
\]
denotes the residual difference relative to a fixed reference point $(f^*,x^*)\in \lip_1\times T$.

Following the framework in~\cite{iclr2025}, $\Psi$ is bounded by separating coordinates $i$ according to the magnitude of the residual at the reference point $(\bar{f},\bar{x})$ and the Lewis weights of $A$.
Contributions from coordinates with large residuals or negligible Lewis weights can be controlled using arguments similar to those in~\cite{iclr2025}.
The core technical challenge is to control the remaining contribution from the set of coordinates $J$ that simultaneously have small residuals and non-negligible Lewis weights.
This term, denoted $\Psi'$, is the focus of the subsequent analysis. After symmetrization, we have
\[
    \E \Psi'\lesssim \E_{\xi, \psi} \frac{1}{\alpha}\sup_{(f,x)\in T} \abs*{\sum_{i\in I} \xi_i Z_i(f,x)},
\]
where $\xi_i$ are independent Rademacher variables and $I\subseteq J$ is the random set of indices $i$ with $\E \psi_i = 1$. As in~\cite{colt}, this term can be controlled via Dudley's integral applied to an Rademacher process indexed by $T$. 

A natural approach is to decouple $f$ and $x$, reducing the Dudley's integral for $T$ to two  Dudley integrals, one over $\pi_1(T)$ and one over $\pi_2(T)$, where $\pi_1(T)$ and $\pi_2(T)$ denote the projections of $T$ onto $\lip_1$ and $\R^d$, respectively.
We note that the decoupling here is more delicate than in~\cite{colt}. In that work, the index set has the product structure $T=\lip_1\times B$ with $B$ a ball in $\R^d$, which allows for a clear separation of $f$ and $x$. However, their result only obtains a constant-factor approximation.

In contrast, when aiming for a $(1+\eps)$-approximation, the only existing argument is the one in~\cite{iclr2025} for the case of a known link function. Adopting this approach leads to index sets $T$ that are not necessarily rectangular; that is, there may not exist sets $X \subseteq \lip_1$ and $Y \subseteq \R^d$ such that $T = X \times Y$. As a result, the decoupling of $f$ and $x$ cannot be carried out at the level of the index set and will instead be introduced halfway through the chain of inequalities used to bound the associated Rademacher process.

The Dudley's integral corresponding to $\pi_2(T)$ has been handled in~\cite{iclr2025}, whereas controlling the $\pi_1(T)$ component is new to this work and constitutes the main technical challenge that arises from the combination of an unknown link function and a general $\ell_p$-norm. Existing arguments from~\cite{colt}, which are tailored to the $p=2$ case, do not generalize easily to other values of $p$.

To address this issue, we develop a simpler and more direct approach for controlling the metric entropy of $\lip_1$. Rather than explicitly discretizing Lipschitz functions as in~\cite{colt}, which leads to delicate and technically involved arguments, we abstract the core difficulty into a covering problem of $\lip_1$ under a family of nonstandard sup-norms. This perspective captures the underlying geometry of the problem while avoiding the technical overhead of explicit discretization and the complicated arguments for controlling set sizes in relation to the discretized functions. The key technical ingredient is the following covering lemma, which may be of independent interest.

\begin{lemma}[Informal version of \Cref{thm:aux_L_infty_norm}]\label{lem:aux_L_infty_norm_intro}
Consider the class $\lip_1$ equipped with a norm defined as the maximum of weighted $L^\infty$-norms over possibly different intervals, 
\[
    \|f\| := \max_{i\in I} \lambda_i \|f\|_{L^\infty([-M_i,M_i])},
\]
where each $\lambda_i\in(0,1)$ is a weight. It then holds that
\[
    \cN(\lip_1, \norm{\cdot}, \eps) \leq \cN\mleft(\lip_1, \norm{\cdot}_{L^\infty([-1,1])}, \frac{\eps}{CM\ln \kappa}\mright),
\]
where $C > 0$ is an absolute constant, $M = \max_{i\in I} \lambda_i M_i$ and $\kappa = (\max_{i\in I}M_i)/M$.
\end{lemma}

We note that the compactness of $\lip_1$ on bounded domains and standard metric entropy bounds for Lipschitz functions ensure finiteness of covering numbers and, for $p>1$, the convergence of Dudley's integral. When $p=1$, the integral does not converge at $0^+$. We therefore control the supremum of the Rademacher process by the sum of a truncated entropy integral and the remaining local oscillation, which results in the same qualitative guarantees as $p>1$.

Finally, as in~\cite{iclr2025}, a bootstrapping argument is used to achieve a $(1+\eps)$-approximation.

\paragraph{Lower Bound.}
The hard instance hides a uniformly random index $I\in[N]$ among $N$ nearly orthogonal unit vectors with coherence at most $\tau$. The matrix $A$ contains a pair of rows for each vector, together with two anchor rows, so we index the rows of $A$ by $(i,+)$ and $(i,-)$, where $i=0,\dots,N$, with $i=0$ corresponding to the anchor rows. The link function $f$ is ReLU, i.e., $f(t) = t^+ = \max\{0,t\}$. For $B=\Theta_p(1/\eps)$, the response vectors are
\[
	b^{(0)}=Be_{0,+},
	\qquad
	b^{(i)}=B(e_{0,+}+e_{i,+}),\quad i\in[N].
\]
Notice that $b^{(i)}$ differs from the baseline $b^{(0)}$ only at the coordinate $(i,+)$.

The main structural step is a short-list guarantee. Given an approximate solution $\widehat x$, define
\[
	H(\widehat x)
	:=\left\{j\in[N]:((A\widehat x)_{j,+})^+\geq\frac{1-\tau}{2}\right\}.
\]
If $N\tau^p\leq\kappa_p\eps^{-p}$, then every solution satisfying the required approximation guarantee \eqref{eqn:known_f_obj} obeys
\[
	i\in H(\widehat x),
	\qquad
	|H(\widehat x)|\leq C_p\eps^{-p}
\]
on input $b^{(i)}$. The first conclusion follows because otherwise the planted coordinate has a large residual. The second follows because every incorrect index in $H(\widehat x)$ contributes a constant amount to the objective while the total objective value is $O(\eps^{-p})$. The anchor rows are used to control the additive term $\norm{Ax^\ast}_p^p$.

To prove the lower bound for randomized algorithms, Yao's minimax theorem reduces the problem to deterministic algorithms under a hard input distribution. We take $I$ uniform on $[N]$ and set $b=b^{(I)}$. Fix a deterministic algorithm with worst-case query complexity $q$ and run it on $b^{(0)}$. Let $Q$ be the set of indices whose coordinates $(j,+)$ it queries, and let $\widehat x^{(0)}$ denote its output. For every $i\notin Q$, the execution on $b^{(i)}$ is identical to the execution on $b^{(0)}$. Hence, if $S$ is the set of inputs on which the algorithm succeeds, then
\[
	S\setminus Q\subseteq H(\widehat x^{(0)}),
	\qquad
	|S|\leq q+C_p\eps^{-p}.
\]
Since success probability $4/5$ gives $|S|\geq4N/5$, taking $N\geq K_p\eps^{-p}$ yields $q=\Omega_p(N)$.

It remains to balance the spherical-code condition $\log N\lesssim d\tau^2$ with the recovery condition $N\tau^p\lesssim_p\eps^{-p}$. Setting $L=\log(2d/\eps)$ and choosing
\[
	\tau\asymp_p\sqrt{\frac{L}{d}},
	\qquad
	N\asymp_p\frac{d^{p/2}}{\eps^pL^{p/2}}
\]
satisfy both conditions when $d\gtrsim_pL$ and give the claimed lower bound.

\section{Preliminaries}

\paragraph{Notation.} 
We use the shorthand $a\vee b=\max\{a,b\}$ and $a\wedge b=\min\{a,b\}$ for any $a,b\in \R$. We also write $[n] = \{1,\dots,n\}$ for positive integers $n$.

We write $X\sim\mathcal{D}$ to denote that a random variable $X$ follows the distribution $\mathcal{D}$. For a scalar $c\in\R$, the notation $X\sim c\mathcal{D}$ means that $X$ has the same distribution as $cY$ for $Y\sim \mathcal{D}$. We use $\Ber(\alpha)$ to denote the Bernoulli distribution with parameter $\alpha$, i.e., the distribution of a random variable that equals $1$ with probability $\alpha$ and $0$ with probability $1-\alpha$. We use $\Bin(n,\alpha)$ to denote the binomial distribution corresponding to $n$ independent Bernoulli trials with success probability $\alpha$.

For nonnegative functions $f$ and $g$, we write $f \lesssim g$ if $f \leq C g$ and $f \gtrsim g$ if $f \geq c g$, where $C$ and $c$ are positive constants. When the constant depends on a variable $a$, we write $f \lesssim_a g$ and $f \gtrsim_a g$. We write $f \asymp g$ when both $f \lesssim g$ and $f \gtrsim g$ hold.

\paragraph{Lewis weights.} Now we recall the definition Lewis weights and some of their basic properties. The forms presented below are from~\cite{iclr2025}.

\begin{definition}[$\ell_p$-Lewis weights]
	Let $A\in \R^{n\times d}$ and $p\geq 1$. For each $i\in [n]$, the $\ell_p$-Lewis weight of $A$ for the $i$-th row is defined to be $w_i$ that satisfies
	\begin{align*}
		w_i(A)
		& =
		(a_i^\top (A^\top W^{1-\frac{2}{p}} A)^{\dagger} a_i)^{\frac{p}{2}}
	\end{align*}
	where $a_i$ is the $i$-th row of $A$ (as a column vector), $W = \diag\{w_1,\dots,w_n\}$ and $\dagger$ denotes the pseudoinverse.
\end{definition}

When the matrix $A$ is clear from context, we simply write $w_i$ instead of $w_i(A)$.

We remark that exact Lewis weights are not needed; constant-factor approximations suffice, increasing the sample complexity by only a constant factor. For simplicity of presentation, however, the algorithms will be stated as if exact Lewis weights were used.

When $p<4$, Cohen and Peng~\cite{Cohen15} show that constant-factor approximate Lewis weights can be computed in $O(\log\log n)$ rounds of leverage-score estimation, with each round taking $\tilde{O}(nd+\poly(d))$ time. When $p>4$, Apers, Gribling, and Sidford~\cite{AGS24} show that they can be computed in $\tilde{O}(d)$ rounds of leverage-score estimation. Consequently, the total running times for approximating the Lewis weights are $\tilde{O}(nd+\poly(d))$ and $\tilde{O}(nd^2+\poly(d))$ in the two regimes, respectively.

\begin{algorithm}[bt]
	\caption{$\texttt{GSM}(k_1,\dots,k_n,\alpha)$ (Generating Sampling Matrix)}\label{alg:sampling}
	\begin{algorithmic}[1]
		\REQUIRE{%
				$n$ integers $k_1,\dots,k_n\geq 0$; a sampling rate $\alpha<1$\strut%
		}%
		\STATE $S \gets$ an $n\times n$ diagonal matrix, initialized to a zero matrix
		\FOR{$i=1,\dots,n$}
		\IF{$k_i > 0$}
		\STATE Generate a binomial variable $N_i\sim \Bin(k_i,\alpha)$
		\STATE $S_{ii} \gets (\frac{N_i}{\alpha k_i})^{\frac{1}{p}}$ \label{step:row_split}
		\ENDIF
		\ENDFOR
		\STATE Return $S$
	\end{algorithmic}
\end{algorithm}

\begin{algorithm}[t]
	\caption{Algorithm for Active Learning}\label{alg:main}
	\begin{algorithmic}[1]
		\REQUIRE{
			\begin{varwidth}[t]{\linewidth}
				a matrix $A \in\mathbb{R}^{n \times d}$ \\
				a query access to the entries of the vector $b\in\mathbb{R}^n$ \\
				an error parameter $\eps$ \\
				a sampling rate $\alpha < 1$ \strut
		\end{varwidth}}
		\STATE Compute the Lewis weights $w_1,\dots,w_n$ of $A$
		\FOR{$i=1,\dots,n$}
		\STATE $k_i \gets \lceil \frac{n \cdot w_i}{d} \rceil$
		\ENDFOR
		\STATE $S \gets \texttt{GSM}(k_1,\dots,k_n,\alpha)$ from \Cref{alg:sampling}
		\STATE Solve the minimization problem 
        \[
        (\hat{f},\hat{x}) := \argmin_{f\in\lip_1,x\in \mathbb{R}^d} \normp{Sf(Ax) - Sb}^p + \eps\normp{Ax}^p
        \]
		\STATE Return $(\hat{f},\hat{x})$
	\end{algorithmic}
\end{algorithm}

\section{Upper Bound}

As previously noted, our approach adopts the framework in~\cite{iclr2025}. The algorithm is presented in \Cref{alg:main}. It is identical to that of~\cite{iclr2025}, with the only difference being that the sketched optimization problem includes minimization over $f$.
We briefly describe the algorithm below. 

For the $i$-th row of $A$, let $k_i = \lceil w_i(A)/(d/n)\rceil$ denote its multiplicity.
Instead of sampling rows directly with probabilities proportional to their Lewis weights $w_i(A)$, we conceptually ``split'' the $i$-th row into $k_i$ identical copies. These copies are sampled independently with a fixed global probability $\alpha = \alpha(d,\eps,p,n)$. The sampled copies are then re-assembled into a single row and rescaled appropriately, yielding a sketched instance that preserves the original objective value.
The algorithm then solves a regularized single-index regression problem of the sketched instance to obtain an approximate solution $(\hat{f},\hat{x})$.

For the purpose of analysis, it is convenient to work directly with the split instance. Specifically, we replace the $i$-th row of $A$ (and the $i$-th entry of $b$) by $k_i$ identical copies, forming an expanded matrix $A'$ (and vector $b'$). Let $m$ be the total number of rows of $A'$ and define an $m\times m$ diagonal matrix
\[
\Lambda = \diag\{\underbrace{k_1^{-1/p},\dots,k_1^{-1/p}}_{\text{$k_1$ times}},\dots,\underbrace{k_n^{-1/p},\dots,k_n^{-1/p}}_{\text{$k_n$ times}}\}.
\]
It is shown in~\cite{iclr2025} that this construction ensures
\begin{enumerate}[topsep=0pt,itemsep=0pt,label=(\roman*)]
\item $m\leq 2n$, 
\item (norm preserving) $\norm{\Lambda(f(A'x)-b')}_p^p = \norm{f(Ax)-b}_p^p$ for all $x\in\R^d$, and 
\item (uniformly bounded Lewis weights) $w_i(\Lambda A')\leq d/n\leq 2d/m$ for all $i\in [m]$. 
\end{enumerate}
Therefore, the single-index regression problem can be equivalently reformulated using $(A',b')$, without changing either the optimal solution or the objective value. For notational simplicity, in the remainder of the paper we continue to write $A$ and $b$, with the understanding that they have been preprocessed to satisfy the properties above.

Once uniformly bounded Lewis weights (of $\Lambda A$) are in place, the algorithm simply samples the coordinates of $b$ independently with probability $\alpha$ and rescales by $\alpha^{-1/p}$. This corresponds to multiplying $f(Ax)$ and $b$ by a diagonal matrix $S$ in which $S_{ii}$ are i.i.d.\ $\alpha^{-1/p}\Ber(\alpha)$ variables.

Analogous to~\cite{iclr2025}, our main theorem in the analysis is the following. 
\begin{restatable}{theorem}{MainErrorBound}  \label{thm:main bound}
Let $p > 1$ be a constant, $A \in \R^{n \times d}$, $\epsilon \in (0,1)$ be sufficiently small, and $\Lambda \in \R^{n \times n}$ be a diagonal matrix with entries $\Lambda_{ii} > 0$ such that $w_i(\Lambda A) \lesssim d/n$ for all $i \in [n]$. Let $v \in \R^n$ be a fixed vector, $(\bar{f}, \bar{x}) \in \lip_1 \times \R^d$ be a fixed reference point, and $V \coloneqq \|\Lambda(\bar{f}(A\bar{x}) - v)\|_p^p$. 

Let $R \geq \|\Lambda A \bar{x}\|_p^p$ and $F \geq V$ be fixed upper bounds. Suppose the set $T$ satisfies 
\[
	\{(\bar{f}, \bar{x})\} \subseteq T \subseteq \{ (f, x) \in \lip_1 \times \mathbb{R}^d : \|\Lambda Ax\|_p^p \leq R \text{ and } \|\Lambda(f(Ax) - \bar{f}(A\bar{x}))\|_p^p \leq F\}.
\]

Let $\alpha\in [0,1]$ and $S$ be an $n \times n$ random diagonal matrix with i.i.d.\ $\alpha^{-1/p}\Ber(\alpha)$ entries. Conditioned on the event that
\begin{gather*}
\norm{ S\Lambda (\bar{f}(A \bar{x})-v) }_{p}^p \lesssim V
\shortintertext{and}
\sup_{(f,x)\in T} \norm{ S\Lambda(f(Ax)-\bar{f}(A\bar{x})) }_{p}^p \leq F,
\end{gather*}
it holds with probability at least $1-\delta$ that
\begin{multline*}
    \sup_{(f,x)\in T} \Big| \left(\norm{S \Lambda(f(Ax)-v)}_p^p - \norm{\Lambda(f(Ax)-v)}_p^p \right) 
    - \left( \norm{S\Lambda(\bar{f}(A\bar{x})-v)}_p^p  - \norm{\Lambda(\bar{f}(A\bar{x})-v)}_p^p\right) \Big| \\
    \leq C \left( \epsilon V + \frac{d^{1 \vee \frac{p}{2}}}{\alpha n} R + \Gamma \ln \frac{n}{\epsilon d} \left( \ln^{\frac{5}{4}} d + \sqrt{\ln \frac{1}{\delta}} \right) \right),
\end{multline*}
where $C$ depends only on $p$, and $\Gamma$ is defined as
\begin{equation} \label{eqn:Gamma}
  \Gamma \coloneqq \
  \begin{cases}
    d^{\frac{1}{2}}(\alpha n)^{-\frac{1}{2}}
    F^{\frac{1}{2}}R^{\frac{1}{2}}, &\quad 1 \leq p \leq 2 \\
    d^{\frac{1}{2}}(\alpha n)^{-\frac{1}{p}}
    F^{1-\frac{1}{p}}R^{\frac{1}{p}}, &\quad p>2.
  \end{cases}
\end{equation} 
\end{restatable}

The proof of \Cref{thm:main bound} is lengthy and technically involved. As outlined in Section~\ref{sec:techniques}, our analysis builds upon the approach in~\cite{iclr2025}, but introduces a key innovation: we bound Dudley's integral over $\lip_1$ by deriving a new bound on the covering number of $\lip_1$ under the supremum of several $L^\infty$-norms over different intervals. The full proof of \Cref{thm:main bound} appears in Appendices~\ref{sec:error_bound_setup} to~\ref{sec:main_error_bound}, culminating in the proof of the theorem in \Cref{sec:main_error_bound}. 
When $p=1$, an analogous version of \Cref{thm:main bound} also holds, with $\log^{\frac{5}{4}}d$ replaced with $\log^{\frac{5}{4}}d + \log n$. The proof is provided in \Cref{sec:p=1}.

Given Theorem~\ref{thm:main bound} (together with its $p=1$ analogue), the mixed-error guarantee for the solution $(\hat{f},\hat{x})$ of the sketched single-index problem follows via a near-identical bootstrapping argument to that of~\cite{iclr2025}. This leads to the following theorem, proved in \Cref{sec:(1+eps)-approximation}.

\begin{restatable}{theorem}{MainUpperBound}  \label{thm:1+eps approx}
  Let $p\geq 1$ be a constant, $A \in \mathbb{R}^{n \times d}$, $b \in \mathbb{R}^n$, $\eps \in (0,1)$ be sufficiently small and $\Lambda$ be an $n \times n$ diagonal matrix satisfying $\Lambda_{ii}>0$ and $w_i(\Lambda A) \lesssim d/n$ for
  all $i$.   
  There exists a randomized algorithm which, with probability at least $0.9$, makes $O( d^{1 \vee \frac{p}{2}}/\eps^{p \vee 2} \poly\log n)$ queries to the entries of $b$ and returns $(\hat{f},\hat{x}) \in \lip_1 \times \mathbb{R}^d$ satisfying the mixed error guarantee \eqref{eqn:main_obj}.
  The hidden constants in the bounds on the number of queries depends on $p$ only.
\end{restatable}

We leave open the question of whether the $\log n$ factors can be removed. In~\cite{iclr2025}, those factors come from Dudley's integral and the authors circumvent them by first using a net argument that leads to a sub-optimal complexity of $\poly(d/\eps)$ queries, which effectively reduces $n$ to $\poly(d/\eps)$. In our case, however, the covering number for Lipschitz functions intrinsically introduces a $\log n$ factor (see \Cref{lem:aux_L_infty_norm_intro}), which a net argument in place of chaining cannot remove. 

\section{Lower Bound}\label{sec:lower}

\paragraph{Basic Setup.}
Let $p>2$, $\alpha = \alpha(p) > 0$ be a constant small enough that is to be determined later and $\eps \in (0, \alpha/2)$. Let $\tau\in(0,1/4)$ to be determined.
By a standard argument, there exist $u_1,\dots,u_N \in \bS^{d-2}$ which satisfy
$\abs{\inner{u_i}{u_j}}\le \tau$ for all $i\neq j$, where $\log N\lesssim \tau^2 d$.
The function $f$ is defined as $f(x) = x^+ := \max\{0, x\}$.  It is clear that $f$ is non-decreasing and $1$-Lipschitz. 

Let $B = \alpha/\eps$, so $B>2$. Define $2N$ (row) vectors in $\R^d$ 
\[
	a_{i,+} = \left( u_i^\top, -\frac{\tau}{B} \right),\quad a_{i,-} = \left(-u_i^\top, -\frac{\tau}{B} \right),\quad i=1,\dots,N
\]
and two additional (row) vectors in $\R^d$
\[
	a_{0,+} = (0,\dots,0,1),\quad a_{0,-} = (0,\dots,0,-1).
\]
Let the matrix $A\in \R^{(2N+2)\times d}$ consist of these $2N+2$ rows, indexed by $(i,\pm)$ for $i=0,\dots,N$.

We define $b^{(0)},\dots,b^{(N)}\in \R^{2N+2}$ as 
\[
	b^{(i)} = 
	\begin{cases}
		B e_{0,+}, & i = 0;\\
		B(e_{0,+} + e_{i,+}), & i > 0,
	\end{cases}
\]
where $e_{0,+}$ and $e_{i,+}$ are standard basis vectors with a single $1$ at the corresponding entry.

Furthermore, we define
\[
	F_i(x) = \norm{(Ax)^+ - b^{(i)}}_p^p,\qquad \OPT_i = \min_x F_i(x)
\]
and $x_i^\ast$ be an arbitrary minimizer of $\min_x F_i(x)$. Note that
\begin{equation}\label{eqn:opt-upper}
	\OPT_i\leq F_i(\begin{bmatrix}
	    u_i \\
        B
	\end{bmatrix}) = (B - (1-\tau))^p \qquad \text{for $i = 1,\dots,N$}.
\end{equation}

For $x\in\R^d$, define
\begin{equation}\label{eqn:H-def}
	H(x) = \left\{ j\in[N]: \langle a_{j,+}, x\rangle^+ \geq \frac{1-\tau}{2} \right\}.
\end{equation}

\begin{lemma}\label{lem:code-reduction}
	There are constants $\kappa_p,K_p,\eps_{0,p}>0$ such that the following holds. Suppose that
	\begin{equation}\label{eq:reduction-conditions}
		N\tau^p\le\kappa_p\eps^{-p},
		\qquad
		N\ge K_p\eps^{-p},
		\qquad
		0<\eps\le\eps_{0,p}.
	\end{equation}
	Let $I$ be uniform on $[N]$ and let $b=b^{(I)}$. Consider a deterministic adaptive algorithm with global worst-case query budget $q$, meaning that it reads at most $q$ entries on every input in $\R^{2N+2}$. If its output $\hat{x}$ satisfies
	\[
		F_I(\hat x)\leq \OPT_I+\eps\bigl(\OPT_I+\norm{Ax_I^\ast}_p^p\bigr)
	\]
	with probability at least $4/5$ over $I$, then $q=\Omega_p(N)$.
\end{lemma}

The proof of Lemma~\ref{lem:code-reduction} is deferred to Appendix~\ref{sec:lower_appendix}.

\begin{theorem}[Large-$d$ lower bound]\label{thm:large-d-code-lb}
	There exist constants $C_{0,p},d_{0,p},\eps_{1,p}>0$ such that the following holds. 
	Suppose that $d\ge d_{0,p}$, $0<\eps\le\eps_{1,p}$, and
	\begin{equation}\label{eqn:large-d-regime}
		d-1 \geq C_{0,p}\log(2d/\eps).
	\end{equation}
	There exist a deterministic function $f\in \lip_1$, an integer $n\geq 1$, a deterministic matrix $A\in\R^{n\times d}$ and a distribution over $b\in\R^n$ such that every deterministic algorithm that outputs $\hat{x}\in\R^d$ which with probability at least $4/5$ over the randomness of $b$ satisfies (\ref{eqn:known_f_obj}) must make 
	\[
		\Omega_p\mleft(\frac{d^{p/2}}{\eps^p (\log(d/\eps))^{p/2}}\mright)
	\]
	to the entries of $b$ in the worst case.
\end{theorem}

\begin{proof}
	Let $\kappa_p,K_p,\eps_{0,p}$ be the constants from Lemma~\ref{lem:code-reduction}.
	Choose constants in the following order: first choose $C_{1,p}$ large, then choose $\eta_p>0$ small, and finally increase $C_{0,p}$ and $d_{0,p}$ and choose $\eps_{1,p}\le\eps_{0,p}$ small enough. Set $m := d-1$,
	\begin{equation}\label{eqn:tau-N-choice}
		L:=\log\frac{2d}{\eps},
		\qquad
		\tau:=C_{1,p}\sqrt{\frac{L}{m}},
		\qquad
		N:=\left\lfloor
		\eta_p\frac{m^{p/2}}{\eps^pL^{p/2}}
		\right\rfloor.
	\end{equation}
	With $C_{0,p}$ large enough, $\tau\le1/4$. Also, $\log N\leq C_{2,p}L$.
	
	Since $m\tau^2=C_{1,p}^2L$, choosing $C_{1,p}$ large enough ensures $\log N\lesssim m\tau^2$. The standard spherical-code argument from the basic setup then supplies $u_1,\dots,u_N\in \bS^{m-1}$ with coherence at most $\tau$. Increasing $d_{0,p}$ if necessary ensures $N\ge m$, so the code may be chosen to span $\R^m$. By construction,
	\[
		N=\Omega_p\!\left(\frac{d^{p/2}}{\eps^pL^{p/2}}\right).
	\]

	Finally,
	\[
		N\tau^p\le\eta_pC_{1,p}^p\eps^{-p}.
	\]
	After $C_{1,p}$ is fixed, choose $\eta_p$ so that $\eta_pC_{1,p}^p\le\kappa_p$. This gives $N\tau^p\le\kappa_p\eps^{-p}$.
	It is also easy to verify that $N\geq K_p\eps^{-p}$ by increasing $C_{0,p}$ if needed. Since $\eps_{1,p}\le\eps_{0,p}$, all the conditions of Lemma~\ref{lem:code-reduction} hold.

	Let $f(t)=t^+$, let $A\in\R^{(2N+2)\times d}$ be the matrix from the basic setup formed from this code, and let $\cD$ be the uniform distribution on $\{b^{(i)}:i\in[N]\}$. For this fixed $f$, the approximation guarantee in \eqref{eqn:known_f_obj} is precisely the event in Lemma~\ref{lem:code-reduction}. That lemma therefore shows that every deterministic algorithm succeeding with probability at least $4/5$ over $b\sim\cD$ makes $\Omega_p(N)$ queries.
\end{proof}

\section*{Acknowledgements}
We thank the anonymous reviewers of ICML 2026 for their suggestions regarding the presentation of the paper.

Following the publication of the ICML proceedings, the authors used ChatGPT 5.6 Sol while developing a new lower bound that also applies to adaptive queries, primarily to explore proof strategies and formulate intermediate results. All mathematical claims and references were independently verified by the authors, who assume full responsibility for the content of this article.

\bibliography{literature}
\bibliographystyle{plain}

\appendix

\section{Additional Preliminaries}

\paragraph{Covering Number and Dudley's Integral.} Suppose that $(X,d)$ is a pseudometric space. The $\eps$-covering number of $X$, denoted by $\cN(X,d,\eps)$, is the minimum $m$ such that there exist $m$ points $x_1,\dots,x_m\in X$ such that for every $x\in X$, there exist $j\in [m]$ such that $d(x,x_j) < \eps$.

Let $(X_t)_{t\in T}$ be a zero-mean subgaussian process with metric $d$ on a set $T$. It is a classical result (see, e.g.~\cite{vershynin}) that 
\begin{equation}\label{eqn:sup_moment_bound1}
\left[\E \left(\sup_{s,t\in T} \abs{X_t - X_s}\right)^\ell\right]^{\frac{1}{\ell}} \lesssim \int_0^\infty \sqrt{\ln \cN(T,d,\eps)} d\eps + \sqrt{\ell} \diam(T,d).
\end{equation}
When the covering number $\sqrt{\ln \cN(T,d,\eps)}$ increases too rapidly as $\eps\to 0^+$, we have 
\begin{equation}\label{eqn:sup_moment_bound2}    
\left[\E \left(\sup_{s,t\in T} \abs{X_t - X_s}\right)^\ell\right]^{\frac{1}{\ell}} 
\lesssim 
\left[\E \left(\sup_{\substack{s,t\in T\\ d(s,t)\leq \eps_0}} \abs{X_t - X_s}\right)^\ell\right]^{\frac{1}{\ell}} +
\int_{\eps_0/2}^\infty \sqrt{\ln \cN(T,d,\eps)} d\eps + \sqrt{\ell} \diam(T,d).
\end{equation}
We shall use the former form for $p>1$ and the latter one for $p=1$.

\paragraph{Rademacher Processes.}
Let $\{\xi_i\}_{i=1}^n$ be independent Rademacher random variables, and let
$T \subset \R^n$. The associated Rademacher process $\{X_t\}_{t \in T}$ on $T$ is defined as
\[
X_t \coloneqq \sum_{i=1}^n \xi_i t_i = \langle \xi, t \rangle .
\]
The canonical (pseudo-)metric of this process is
\[
d(s,t) \coloneqq ( \E(X_s - X_t)^2 )^{1/2} = \| s - t \|_2 , \qquad s,t \in T.
\]
Moreover, the process $\{X_t\}_{t \in T}$ is a subgaussian process with respect to the metric $d$.

\paragraph{Lewis weights.} First we present a lemma on the basic properties of Lewis weights from~\cite{iclr2025}.

\begin{lemma}\label{lem: lewis weight properties}
    Suppose that $A \in \mathbb{R}^{n \times d}$ has full column rank and Lewis weights $w_1,\dots,w_n$. Let $W = \diag\{{w_1,\dots,w_n}\}$. The following properties hold.

\begin{enumerate}[label=(\arabic*),topsep=0pt,itemsep=0pt]
    \item $\sum_i w_i = d$;
    \item There exists a matrix $U \in \R^{n \times d}$ such that
    \begin{enumerate}[topsep=0pt,itemsep=0pt,parsep=0pt]
        \item the column space of $U$ is the same as that of $A$;
        \item $w_i = \norm{U_{i,\ast}}_2^p$, where $U_{i,\ast}$ denotes the $i$-th row of $U$;
        \item $W^{\frac{1}{2}- \frac{1}{p}}U$ has orthonormal columns;
    \end{enumerate}
    \item It holds for all vectors $u$ in the column space of $A$ that $\norm{W^{\frac{1}{2}- \frac{1}{p}}u}_2 \leq d^{\frac{1}{2}- \frac{1}{2 \vee p}} \norm{u}_p$.
    \item It holds for all vectors $u$ in the column space of $A$ that $\abs{u_i} \leq d^{\frac{1}{2}- \frac{1}{2 \vee p}} w_i^{\frac{1}{p}}\norm{u}_p$.
\end{enumerate}
\end{lemma}

We present some properties regarding $\Lambda_{ii}(f_1(Ax_1) - f_2(Ax_2))_i$.

\begin{lemma} \label{lem: bound for (lambda Ax)_i}
    For any $R \geq \norm{\Lambda A \bar{x}}_p^p$, let $T = \{ x \in \mathbb{R}^d : \norm{\Lambda A x}_p^p \leq R\}$. Also, let $w_1,\ldots,w_n$ be Lewis weights of $\Lambda A$. It holds for all $x, \bar{x} \in T$, $f,\bar{f}\in \lip_1$ and $i \in [n]$ that
    \begin{gather*}
        \abs{\Lambda_{ii} (f(Ax) - f(A\bar{x}))_i} \leq 2d^{\frac{1}{2}- \frac{1}{2 \vee p}} w_i^{\frac{1}{p}} R^{\frac{1}{p}}, \\
        \abs{\Lambda_{ii} (f(Ax) - \bar{f}(Ax))_i} \leq 2d^{\frac{1}{2}- \frac{1}{2 \vee p}} w_i^{\frac{1}{p}} R^{\frac{1}{p}}.
    \end{gather*}
    Consequently, by the triangle inequality,
    \[
        \abs{\Lambda_{ii} (f(Ax) - \bar{f}(A\bar{x}))_i} \leq 4d^{\frac{1}{2}- \frac{1}{2 \vee p}} w_i^{\frac{1}{p}} R^{\frac{1}{p}}.
    \]
\end{lemma}

\begin{proof}
    By the Lipschitz condition and by Lemma~\ref{lem: lewis weight properties}(4),
    \[
        \abs{\Lambda_{ii} (f(Ax) - f(A\bar{x}))_i} \leq \abs{\Lambda_{ii}(Ax - A \bar{x})_i} 
        \leq d^{\frac{1}{2}- \frac{1}{2 \vee p}} w_i^{\frac{1}{p}} \norm{\Lambda Ax - \Lambda A \bar{x}}_p.
    \]
    Since both $x, \bar{x}\in T$, we have
    \[
        \norm{\Lambda Ax - \Lambda A \bar{x}}_p \leq \norm{\Lambda Ax}_p + \norm{\Lambda A \bar{x}}_p \leq 2 R^{\frac{1}{p}}.
    \]
    This proves the first inequality.

    For the second inequality, since $f,\bar{f}\in \lip_1$, we know that $f-\bar{f}\in \lip_2$ and, by Lemma~\ref{lem: lewis weight properties}(4) again,
    \[
        \abs{\Lambda_{ii} (f(Ax) - \bar{f}(Ax))_i} \leq 2\abs{\Lambda_{ii}(Ax)_i} \leq 2d^{\frac{1}{2}- \frac{1}{2 \vee p}} w_i^{\frac{1}{p}}\norm{\Lambda A x}_p \leq 2d^{\frac{1}{2}- \frac{1}{2 \vee p}} w_i^{\frac{1}{p}}R^{\frac{1}{p}}.
    \]
    This proves the second inequality.
\end{proof}

\paragraph{Moment bounds for the binomial variables} We shall need an upper bound on higher moments of binomial variables.
\begin{lemma}[Ahle~\cite{ahle2022binomial}]
\label{lem:binomial_moment}
Let $X \sim \Bin(n,p)$ be a binomial random variable with mean $\mu = np$.
Then for every integer $\ell \ge 1$,
\[
\E X^\ell \leq \left(\mu + \frac{\ell}{2} \right)^\ell.
\]
\end{lemma}

\section{Formulation of the Error Bound and Analysis for Upper Bound}\label{sec:error_bound_setup}

Suppose that $T$ is a subset of $\lip_1\times \R^d$. Let $\pi_1$ denote the projection on $\lip_1$ and $\pi_2$ the projection on $\R^d$.

Recall that we need to bound the following error:
\begin{equation}
  \label{error bound}
  \Psi := \sup_{(f,x)\in T} \abs[\big]{ ( \norm{ S
    \Lambda(f(Ax)-v) }_p^p - \norm{ \Lambda(f(Ax)-v) }_p^p ) 
     - (\norm{S \Lambda(\bar{f}(A\bar{x})-v) }_p^p  - \norm{\Lambda(\bar{f}(A\bar{x})-v) }_p^p) },
\end{equation}
which can be rewritten as
\begin{equation}\label{eqn:Psi_second}
  \Psi = \sup_{(f,x)\in T} \abs*{ \sum_{i = 1}^n 
    (S_{ii}^p - 1) (\abs{(f(Ax)-v)_i}^p - \abs{(\bar{f}(A\bar{x})-v)_i}^p) }.
\end{equation}

Similar to the approach in~\cite{iclr2025}, we partition the index set $[n]$ into three disjoint sets $J$, $G\setminus J$ and $G^c$, where
\begin{gather*}
    G \coloneqq \left\{i\in [n]:
    |\Lambda_{ii}(\bar{f}(A\bar{x})-v)_i| \leq \frac{d^{\frac{1}{2}-
    \frac{1}{2 \vee p}}w_i^{\frac{1}{p}}R^{\frac{1}{p}}}{\eps} \right\},
    \quad 
    J \coloneqq \left\{i\in G: w_i > \frac{\eps^p d}{n^2}\right\}.
\end{gather*}
Correspondingly, we split the summation over $[n]$ in \eqref{eqn:Psi_second} into three sums, and consequently, 
\[
    \Psi\leq \Psi_1 + \Psi_2 + \Psi_3,
\]
where
\begin{align*}
    \Psi_1 &=  \sup_{(f,x)\in T} \abs*{ \sum_{i\not\in G} (S_{ii}^p - 1) \Lambda_{ii}^p (\abs{(f(Ax)-v)_i}^p - \abs{(\bar{f}(A\bar{x})-v)_i}^p) } \\
    \Psi_2 &=  \sup_{(f,x)\in T} \abs*{ \sum_{i\in G\setminus J} (S_{ii}^p - 1) \Lambda_{ii}^p (\abs{(f(Ax)-v)_i}^p - \abs{(\bar{f}(A\bar{x})-v)_i}^p) } \\
    \Psi_3 &=  \sup_{(f,x)\in T} \abs*{ \sum_{i\in J} (S_{ii}^p - 1) \Lambda_{ii}^p (\abs{(f(Ax)-v)_i}^p - \abs{(\bar{f}(A\bar{x})-v)_i}^p) }
\end{align*}

Bounding $\Psi_1$ and $\Psi_2$ is relatively straightforward; the arguments follow those in~\cite{iclr2025}, though they are slightly more involved due to the varying functions $f$. 

\begin{lemma}[Bounding $\Psi_1$]
  For any $R \geq \norm{\Lambda A\bar{x}}_p^p$ and $\eps >
  0$, let $T$ be a set such that $\{ \bar{x}\} \subseteq T \subseteq \{x
  \in \mathbb{R}^d : \norm{\Lambda A x}_p^p \leq R \}$. Suppose that $S$ is an $n$-by-$n$ diagonal matrix with
  non-negative diagonal entries and
  \[
    \norm{S\Lambda(\bar{f}(A\bar{x})-v)}_p^p \lesssim V,
  \]
  where $V = \norm{\Lambda \bar{f}(A\bar{x})-v}_p^p$. Then we have
  \[
    \Psi_1 \lesssim \eps V.
  \]
\end{lemma}

\begin{proof}
  By the triangle inequality,
  \begin{align*}
    &\abs*{ \sum_{i\not\in G} (S_{ii}^p - 1) \Lambda_{ii}^p (\abs{(f(Ax)-v)_i}^p - \abs{(\bar{f}(A\bar{x})-v)_i}^p)}\\
    \leq\,& \sum_{i \notin G} (S_{ii}^p + 1) \Lambda_{ii}^p\,
    \abs[\big]{\, |(f(Ax)-v)_i|^p - |(\bar{f}(A\bar{x})-v)_i|^p \,}.
  \end{align*}
  Using the inequality $\abs{ |a|^p - |b|^p } \leq p |a-b|\,
  (|a|^{p-1} + |b|^{p-1})$, we have
  \begin{align*}
    & \Lambda_{ii}^p \abs[\big]{\, |(f(Ax)-v)_i|^p -
    |(\bar{f}(A\bar{x})-v)_i|^p \, } \\
    \leq\, & p\, \abs{\Lambda_{ii}(f(Ax) - \bar{f}(A\bar{x}))_i } \, (|\Lambda_{ii}(f(Ax) - v)_i|^{p-1} +
    |\Lambda_{ii}(\bar{f}(A\bar{x}) - v)_i|^{p-1}).
  \end{align*}
  Hence, for $i \notin G$, by Lemma~\ref{lem: bound for (lambda Ax)_i} and the definition of $G$, it holds that
  \[
    \abs{ \Lambda_{ii}(f(Ax) - \bar{f}(A\bar{x}))_i } 
    \leq 4d^{\frac{1}{2}- \frac{1}{2 \vee p}} w_i^{\frac{1}{p}} R^{\frac{1}{p}}
    \leq 4\eps \abs{\Lambda_{ii}(\bar{f}(A\bar{x})-v)_i }
  \]
  and so
  \[
    \abs{\Lambda_{ii}(f(Ax) - v)_i } 
    \leq \abs{ \Lambda_{ii}(f(Ax) - \bar{f}(A\bar{x}))_i } +
    \abs{\Lambda_{ii}(\bar{f}(A\bar{x}) - v)_i  } 
    \leq (1 + 4\eps)\abs{\Lambda_{ii}(\bar{f}(A\bar{x}) - v)_i}.
  \]
  It follows that
  \[
    \Lambda_{ii}^p \abs{ (f(Ax)-v)_i|^p -
    |(\bar{f}(A\bar{x})-v)_i|^p } \lesssim \eps
    \abs{\Lambda_{ii}(\bar{f}(A\bar{x}) - v)_i  }^p.
  \]
  Therefore,
  \begin{align*}
    \abs*{ \sum_{i\not\in G} (S_{ii}^p - 1) \Lambda_{ii}^p (\abs{(f(Ax)-v)_i}^p - \abs{(\bar{f}(A\bar{x})-v)_i}^p)} &\lesssim \eps \sum_{i \notin G} (S_{ii}^p + 1)\abs{\Lambda_{ii}(\bar{f}(A\bar{x}) - v)_i  }^p \\
    &\leq \eps \sum_{i = 1}^n (S_{ii}^p + 1)\abs{\Lambda_{ii}(\bar{f}(A\bar{x}) - v)_i}^p 
    \lesssim \eps V. \qedhere
  \end{align*}
\end{proof}

\begin{lemma}[Bounding $\Psi_2$]
  For any $R \geq \norm{\Lambda A\bar{x}}_p^p$ and $\eps >
  0$, let $T$ be a set that $\{ \bar{x}\} \subseteq T \subseteq
  \{ x \in \mathbb{R}^d : \norm{\Lambda Ax}_p^p \leq R\}$. 
  Suppose $S \in \mathbb{R}^{n \times n}$ is a diagonal
  matrix whose entries satisfy $0 \leq S_{ii}^p \leq
  \frac{1}{\alpha}$ for any $\alpha > 0$. Then
  \[
    \Psi_2 \lesssim \frac{d^{1 \vee \frac{p}{2}}}{\alpha n} R.
  \]
\end{lemma}

\begin{proof}
  By the triangle inequality,
  \begin{align*}
    & \abs*{ \sum_{i\in G\setminus J} (S_{ii}^p - 1) \Lambda_{ii}^p (\abs{(f(Ax)-v)_i}^p - \abs{(\bar{f}(A\bar{x})-v)_i}^p) }\\
    \leq\, & \sum_{i \in G \setminus J} (S_{ii}^p + 1)\Lambda_{ii}^p\, \abs[\big]{|(f(Ax)-v)_i|^p - |(\bar{f}(A\bar{x})-v)_i|^p }
  \end{align*}
  It holds for $i \in G$ and $x\in T$ that
  \begin{align*}
    |\Lambda_{ii}(f(Ax)-v)_i| &\leq
    |\Lambda_{ii}(f(Ax)-\bar{f}(A\bar{x}))_i| +
    |\Lambda_{ii}(\bar{f}(A\bar{x})-v)_i|  \\
    &\leq \biggl(4 + \frac{1}{\eps} \biggl) d^{\frac{1}{2}-\frac{1}{2 \vee p}}w_i^{\frac{1}{p}}R^{\frac{1}{p}}, & \text{(by \Cref{lem: bound for (lambda Ax)_i})}
  \end{align*}
  When one further has $i\not\in J$,
  \[
    \abs{\Lambda_{ii}(f(Ax)-v)_i}^p 
    \lesssim 
    \frac{d^{\frac{p}{2}\vee 1 - 1}}{\eps^p}w_i R \lesssim \frac{d^{1 \vee \frac{p}{2}}}{n^2}R.
    \]
  Since $S_{ii}^p \leq \frac{1}{\alpha}$, we have
  \begin{align*}
  &\sup_{(f,x)\in T} \abs*{\sum_{i \in G \setminus J}
    (S_{ii}^p - 1)\Lambda_{ii}^p \biggl(|(f(Ax)-v)_i|^p -
  |(\bar{f}(A\bar{x})-v)_i|^p \biggl) } \\
  \lesssim\, & \sum_{i \in G \setminus J} \left( \frac{1}{\alpha} +1
  \right)\frac{d^{1 \vee \frac{p}{2}}}{n^2}R \\
  \lesssim\, &\frac{d^{1 \vee \frac{p}{2}}}{\alpha n} R. \qedhere
  \end{align*}
\end{proof}

Next we bound $\Psi_3$, which is the core technical difficulty. Recall that
\begin{align*}
\Psi_3 = \sup_{(f,x)\in T} \abs*{
\sum_{i \in J} (S_{ii}^p - 1)\Lambda_{ii}^p \, (|(f(Ax)-v)_i|^p - |(\bar{f}(A\bar{x})-v)_i|^p) }.
\end{align*}
By the symmetrization trick, we have
\begin{equation}\label{eqn:Psi_3_moment_original}
\E_S \Psi_3^\ell \leq 2^\ell \E_{\xi,S} \left(
  \sup_{(f,x)\in T} \abs*{ 
    \sum_{i \in J} \xi_i \cdot S_{ii}^p \Lambda_{ii}^p (|(f(Ax)-v)_i|^p -
|(\bar{f}(A\bar{x})-v)_i|^p ) } \right)^\ell.
\end{equation}
where $\xi$ has the same dimension as $J$ and whose entries are
independent Rademacher random variables.

For every $x \in \mathbb{R}^d$ and $f\in \lip_1$, we define $Z(f,x)\in\R^n$ coordinatewise as
\[
Z_i(f,x) \coloneqq \Lambda_{ii}^p |(f(Ax)-v)_i|^p.
\]
Let $I \subseteq J$ be the set of indices $i$ such
that $S_{ii}^p = \frac{1}{\alpha}$. We can then write \eqref{eqn:Psi_3_moment_original} as
\begin{equation}\label{eqn:Psi_3_moment_rewritten}
\E_S \Psi_3^\ell 
\leq 
2^\ell \E_{\xi,S} \left(\sup_{(f,x)\in T}\abs*{\frac{1}{\alpha} \inner{\xi_I}{Z_I(f,x)} - \inner{\xi_I}{Z_I(\bar{f},\bar{x})} }\right)^\ell.
\end{equation}

Next we condition on $S$.  Define for $(f_1,x_1), (f_2,x_2)\in \lip_1\times \R^d$ the following pseudometric.
\begin{align*}
D_2((f_1,x_1),(f_2,x_2)) &\coloneqq 
\left( \sum_{i \in I}
  \abs[\big]{Z_i(f_1, x_1) - Z_i(f_2,x_2)}^2 \right)^{\frac{1}{2}} \\
&=\left( \sum_{i \in I}
  \abs[\big]{|\Lambda_{ii}(f_1(Ax_1)-v)_i|^p - |\Lambda_{ii}(f_2(Ax_2)
-v)_i|^p }^2 \right)^{\frac{1}{2}}. 
\end{align*}

Since $(\bar{f},\bar{x})\in T$, the $\ell$-th moment on the right-hand side of \eqref{eqn:Psi_3_moment_rewritten} can be upper bounded using Dudley's integral \eqref{eqn:sup_moment_bound1} as
\begin{equation}\label{eqn:Psi_3_moment_dudley}
    \E_\xi \left(\sup_{(f,x)\in T}\abs[\big]{ \inner{\xi_I}{Z_I(f,x)} - \inner{\xi_I}{Z_I(\bar{f},\bar{x})} }\right)^\ell \\
    \leq 
    C^\ell \left(\int_0^\infty \sqrt{\ln\cN(T, D_2, \eps)} d\eps + \sqrt\ell \diam(T, D_2)\right)^\ell
\end{equation}
for some absolute constant $C > 0$. We shall bound the Dudley's integral and the diameter.

\section{Upper Bounding $D_2$}
We shall upper bound $D_2$ on $T$ using simpler pseudometrics. Define for $f_1,f_2\in\lip_1$ and $x,x'\in\R^d$ the following pseudometrics.
\begin{gather*}
d_T(f_1,f_2) \coloneqq \sup_{x\in \pi_2(T)} \norm{f_1(Ax) - f_2(Ax)}_{I,\infty} \\
\rho_T(x_1,x_2) \coloneqq \sup_{f\in \pi_1(T)} \norm{f(Ax_1) - f(Ax_2)}_{I,\infty}
\end{gather*}

\begin{lemma}\label{lem:D_2_bound}
Let $\{(\bar{f}, \bar{x})\} \subseteq T \subseteq \{(f, x) \in \lip_1\times \R^d: \norm{\Lambda Ax}_p^p \leq R \text{ and } \norm{\Lambda (\bar{f}(A\bar{x})-v)}_p^p\leq F\}$. 
If $I \subseteq J$ satisfies
\[  
    \norm{\Lambda (\bar{f}(A \bar{x})-v)}_{I,p}^p \lesssim \alpha
    \norm{\Lambda (\bar{f}(A\bar{x})-v)}_p^p
    \quad\text{and}\quad
    \sup_{(f,x)\in T} \norm{\Lambda(f(Ax)-\bar{f}(A\bar{x})}_{I,p}^p \lesssim \alpha F,
\]
then
\[
D_2((f_1, x_1), (f_2, x_2)) \lesssim (\alpha F)^\theta\left( (d_T(f_1, f_2))^\phi + (\rho_T(x_1, x_2))^\phi\right),
\]
where
\[
\theta = \left(1-\frac{1}{p}\right)\vee \frac{1}{2},\quad \phi = \frac{p}{2}\wedge 1.
\]
\end{lemma}
\begin{proof}
    Recall that
    \[
        D_2((f_1,x_1),(f_2,x_2))^2 = \sum_{i \in I}
            \abs[\big]{|\Lambda_{ii}(f_1(Ax_1)-v)_i|^p - |\Lambda_{ii}(f_2(Ax_2)-v)_i|^p }^2.
    \]
    For notational convenience, we write $u_1 = f_1(Ax_1)$, $u_2 = f_2(Ax_2)$ and $\bar{u} = \bar{f}(A\bar{x})$. Using the fact that $\abs{|a|^p - |b|^p } \leq p|a - b| (|a|^{p-1} + |b|^{p-1})$, we have
    \[
        |\Lambda_{ii}(u_1-v)_i|^p - |\Lambda_{ii}(u_2-v)_i|^p
        \leq p \abs{\Lambda_{ii}(u_1-u_2)_i}\,
        (|\Lambda_{ii}(u_1-v)_i|^{p-1} + |\Lambda_{ii}(u_2-v)_i|^{p-1}).
    \]
    It then follows that
    \[
        \sum_{i \in I} \abs{|\Lambda_{ii}(u_1-v)_i|^p - |\Lambda_{ii}(u_2-v)_i|^p}^2 \\
        \lesssim
        \sum_{i\in I} \abs{\Lambda_{ii}(u_1 - u_2)_i }^2 
        \left( |\Lambda_{ii}(u_1 - v)_i|^{2p-2} + |\Lambda_{ii}(u_2 - v)_i|^{2p-2} \right).
    \]
    We shall proceed in two ways, depending on the value of $p$.
    \begin{mycase}
    \case $1 \leq p \leq 2$. We begin with
    \begin{align*}
        &\sum_{i\in I} \abs{\Lambda_{ii}(u_1 - u_2)_i }^2 
        \left( |\Lambda_{ii}(u_1 - v)_i|^{2p-2} + |\Lambda_{ii}(u_2 - v)_i|^{2p-2} \right) \\
        \leq \,&
        \left(\max_{i\in I} \abs{\Lambda_{ii}(u_1 - u_2)_i}^p\right)
        \sum_{i\in I} \abs{\Lambda_{ii}(u_1 - u_2)_i }^{2-p}
        \left( |\Lambda_{ii}(u_1 - v)_i|^{2p-2} + |\Lambda_{ii}(u_2 - v)_i|^{2p-2} \right).
    \end{align*}
    The sum can then be upper bounded as 
    \begin{align*}
        &\sum_{i \in I} \abs{\Lambda_{ii}(u_1-u_2)_i}^{2-p} 
        \left( |\Lambda_{ii}(u_1-v)_i|^{2p-2} + |\Lambda_{ii}(u_2-v)_i|^{2p-2}\right) \\
        \lesssim\, & \sum_{i \in I} \Lambda_{ii}^p \abs{(u_1-u_2)_i }^{2-p} \max \left\{|(u_1-v)_i|^{2p-2}, |(u_2-v)_i|^{2p-2}  \right\} \\      \lesssim\, & \left(\sum_{i \in I} \Lambda_{ii}^p \abs{(u_1-u_2)_i }^{p} \right)^{\frac{2-p}{p}} \left( \sum_{i \in I} \Lambda_{ii}^p \max \left\{|(u_1-v)_i|^{p},  |(u_2-v)_i|^{p}\right\}\right)^{\frac{2p-2}{p}} \\
        \leq\, & \norm{\Lambda(u_1-u_2)}_{I,p}^{2-p}
        \left(\norm{\Lambda(u_1-v)}_{I,p}^p + \norm{\Lambda(u_2-v)}_{I,p}^p\right)^{\frac{2p-2}{p}}
    \end{align*}
    Since $u_j - v = u_j - \bar{u} + \bar{u} - v$ ($j = 1, 2$), by the triangle inequality,
    \begin{align*}
      \norm{\Lambda(u_j-v)}_{I,p} &\leq \norm{\Lambda(u_j-\bar{u})}_{I,p} + \norm{\Lambda(\bar{u}-v)}_{I,p} \\
      &\lesssim (\alpha F)^{\frac{1}{p}} + \left(\alpha \norm{\Lambda(\bar{u}-v)}_{p}^p\right)^\frac{1}{p} \\
      &\lesssim (\alpha F)^{\frac{1}{p}}.
    \end{align*}
    Moreover,
    \[
      \norm{\Lambda(u_1 - u_2)}_{I,p} \leq
      \norm{\Lambda(u_1-\bar{u})}_{I,p} +
      \norm{\Lambda(u_2-\bar{u})}_{I,p} 
      \lesssim (\alpha F)^{\frac{1}{p}}.
    \]
    Putting everything together, we have
    \[
      D_2((f_1,x_1),(f_2,x_2))^2 \lesssim \norm{\Lambda(u_1-u_2)}_{I,\infty}^p (\alpha F)^{\frac{2-p}{p}} (\alpha F)^{\frac{2p-2}{p}}
      \leq \alpha F \norm{\Lambda(u_1-u_2)}_{I,\infty}^p.
    \]

    \case $p > 2$. Note that for any vector $z$, we have $|z_i|^{2p-2} \leq
    |z_i|^p \norm{z}_\infty^{p-2}$. Hence,
    \begin{align*}
      & \sum_{i \in I} \abs{ \Lambda_{ii}(u_1-u_2)_i
      }^{2} \left( |\Lambda_{ii}(u_1-v)_i|^{2p-2} +
      |\Lambda_{ii}(u_2-v)_i|^{2p-2} \right) \\
      \leq\, & \sum_{i \in I}  \norm{\Lambda(u_1-u_2)}_{I,\infty}^{2} \left(
        |\Lambda_{ii}(u_1-v)_i|^p \norm{\Lambda(u_1-v)}_{I,\infty}^{p-2} +
        |\Lambda_{ii}(u_2-v)_i|^p \norm{\Lambda(u_2-v)}_{I,\infty}^{p-2} \right) \\
      =\, & \norm{\Lambda(u_1-u_2)}_{I,\infty}^{2}
      \left( \norm{\Lambda(u_1-v)}_{I,p}^p \norm{\Lambda(u_1-v)}_{I,\infty}^{p-2} + \norm{\Lambda(u_2-v)}_{I,p}^p \norm{\Lambda(u_2-v)}_{I,\infty}^{p-2} \right) \\
      \leq\, & (\alpha F) \norm{\Lambda(u_1-u_2)}_{I,\infty}^{2} \left(\norm{\Lambda(u_1-v)}_{I,\infty}^{p-2} + \norm{\Lambda(u_2-v)}_{I,\infty}^{p-2} \right) \\
      \leq\, & (\alpha F) \norm{\Lambda(u_1-u_2)}_{I,\infty}^{2} \left(\norm{\Lambda(u_1-v)}_{I,p}^{p-2} + \norm{\Lambda(u_2-v)}_{I,p}^{p-2} \right).
    \end{align*}
    Recall that $\norm{\Lambda(u_1-v)}_{I,p} \lesssim (\alpha
    F)^{\frac{1}{p}}$ and $\norm{\Lambda(u_2-v)}_{I,p} \lesssim (\alpha
    F)^{\frac{1}{p}}$. It then follows that
    \[
      D_2((f_1,x_1),(f_2,x_2))^2 \lesssim (\alpha F) (\alpha
      F)^{\frac{p-2}{p}} \norm{\Lambda(u_1-u_2)}_{I,\infty}^{2} \\
      \leq (\alpha F)^{2 - \frac{2}{p}} \norm{\Lambda(u_1-u_2)}_{I,\infty}^{2}.
    \]
    \end{mycase}
    Combining both cases, we obtain that
    \begin{equation}\label{eqn:D_2_intermediate}
        D_2((f_1,x_1),(f_2,x_2)) \lesssim (\alpha F)^{\frac{1}{2}\vee(1-\frac{1}{p})}\norm{f_1(Ax_1)-f_2(Ax_2)}_{I,\infty}^{\frac{p}{2}\wedge 1}.
    \end{equation}
    By the triangle inequality,
    \begin{equation}\label{eqn:D_2_aux_last_step}
        \norm{f_1(Ax_1)-f_2(Ax_2)}_{I,\infty}
        \leq
        \norm{f_1(Ax_1)-f_2(Ax_1)}_{I,\infty} + \norm{f_2(Ax_1)-f_2(Ax_2)}_{I,\infty}.
    \end{equation}
    Plugging \eqref{eqn:D_2_aux_last_step} into \eqref{eqn:D_2_intermediate} yields immediately that
    \[
        D_2((f_1,x_1),(f_2,x_2)) \lesssim (\alpha F)^{\frac{1}{2}\vee(1-\frac{1}{p})} \left[ \sup_{x\in \pi_2(T)} \norm{f_1(Ax) - f_2(Ax)}_{I,\infty}^{\frac{p}{2}\wedge 1} 
         + \sup_{f\in \pi_1(T)} \norm{f(Ax_1) - f(Ax_2)}_{I,\infty}^{\frac{p}{2}\wedge 1} \right]
    \]
    as desired.
\end{proof}

\begin{lemma}\label{lem:main_dudley_bound}
Under the assumptions of \Cref{lem:D_2_bound}, it holds that
\begin{equation}\label{eqn:main_dudley_bound}
  \int_0^\infty \sqrt{\ln \cN(T, D_2, \eps)} d\eps
  \lesssim
  (\alpha F)^\theta \left[
  \int_0^\infty \sqrt{\ln \cN(\pi_1(T), d_T^\phi, \eps)}
  d\eps +  \int_0^\infty\sqrt{\ln \cN(\pi_2(T), \rho_T^\phi, \eps)} d\eps
  \right]
\end{equation}
and
\[
    \diam(T, D_2) \leq (\alpha F)^\theta \left[\diam(\pi_1(T), d_T^\phi) + \diam(\pi_2(T), \rho_T^\phi)\right].
\]
\end{lemma}

\begin{proof}
\Cref{lem:D_2_bound} implies that
\begin{gather*}
\diam(T, D_2) \leq (\alpha F)^\theta \left[\diam(\pi_1(T), d_T^\phi) + \diam(\pi_2(T), \rho_T^\phi)\right] \\
\shortintertext{and}
\cN(T, D_2, \eps) 
\leq \cN\left(\lip_1, (\alpha F)^\theta d_T^\phi, \frac{\eps}{2}\right)\cdot 
\cN\left(T, (\alpha F)^\theta \rho_T^\phi, \frac{\eps}{2}\right).
\end{gather*}
It follows that
\[
    \sqrt{\ln \cN(T, D_2, \eps)} 
    \leq \sqrt{\ln \cN\left(\lip_1, (\alpha F)^\theta d_T^\phi, \frac{\eps}{2}\right)} +  
    \sqrt{\ln \cN\left(T, (\alpha F)^\theta \rho_T^\phi, \frac{\eps}{2}\right)}
\]
and
\begin{align*}
  \int_0^\infty \sqrt{\ln \cN(T, D_2, \eps)} d\eps 
  &\lesssim \int_0^\infty \sqrt{\ln \cN\mleft(\pi_1(T), (\alpha F)^\theta d_T^\phi, \frac{\eps}{2}\mright)} d\eps + \int_0^\infty \sqrt{\ln \cN\mleft(\pi_2(T), (\alpha F)^\theta \rho_T^\phi, \frac{\eps}{2}\mright)} d\eps \\
  &\leq 2(\alpha F)^\theta \left(\int_0^\infty \sqrt{\ln \cN(\pi_1(T), d_T^\phi, \eps)} d\eps + \int_0^\infty \sqrt{\ln \cN(\pi_2(T), \rho_T^\phi, \eps)} d\eps\right). \qedhere
\end{align*}
\end{proof}

We shall bound the two Dudley's integrals in the right-hand side of \eqref{eqn:main_dudley_bound} separately. Li and Tai~\cite{iclr2025} computed the second integral regarding $\pi_2(T)$ for the problem with a known $f$. We summarize the result in the following lemma.

\begin{lemma}[Li and Tai~\cite{iclr2025}]\label{lem:dudley_integral_T}
Under the assumptions of \Cref{lem:D_2_bound}, it holds that
\begin{equation}\label{eqn:dudley_integral_T}
    (\alpha F)^\theta \int_0^\infty \sqrt{\ln \cN\mleft(\pi_2(T),\rho_T^\phi,\eps\mright)} d\eps \lesssim \alpha \Gamma (\ln^{\frac{5}{4}} d) \sqrt{\ln \frac{n}{\eps d}}, \quad (\alpha F)^\theta \diam(\pi_2(T),\rho_T^\phi) \lesssim \alpha \Gamma,
\end{equation}
where $\Gamma$ is defined as in~\eqref{eqn:Gamma}.
\end{lemma}

\section{Bounds Pertaining to $\lip_1$}
It remains to upper bound the first integral regarding $\lip_1$. 
First, we prove an auxiliary result that bounds the covering
number of $\text{Lip}_1$ under $\norm{\cdot}$, where $\norm{\cdot}$ is the supremum of a collection of $L^\infty$-norms on different intervals, by the covering number under the $L^\infty$-norm on $[-1,1]$.

\begin{theorem}\label{thm:aux_L_infty_norm}
    Suppose that each $i$ in a finite index set $I$ corresponds to a bound threshold $M_i > 0$ and a weight $\lambda_i\in (0,1)$. Consider the norm $\norm{\cdot}$ on $\lip_1$ defined as
    \[
        \norm{f} := \max_{i\in I} \lambda_i \norm{f}_{L^\infty([-M_i,M_i])}.
    \]
    Let $M = \max_{i\in I} \lambda_i M_i$ and $\kappa = (\max_{i\in I} M_i)/M$. It holds for $\eps\in(0,1)$ that
    \[
        \cN(\lip_1, \norm{\cdot}, \eps) \leq \cN\left(\lip_1, \norm{\cdot}_{L^\infty([-1,1])}, \frac{\eps}{2M(\ln\kappa + 1)}\right).
    \]
\end{theorem}

\begin{proof}
  Let $\beta = 1/(1 + \ln\kappa)$.
  For each $f \in \text{Lip}_1$, define $Tf: [-1,1]\to \mathbb{R}$ as
  \begin{align*}
    Tf(x) \coloneqq
    \begin{cases}
      \frac{1}{M}f\big(\frac{M}{\beta}x\big), &\quad |x| \leq \beta \\
      f\big(\sgn(x)\cdot \kappa^{\frac{|x|-\beta}{1-\beta}}M\big)/\big(\kappa^{\frac{|x|-\beta}{1-\beta}}M\big),&
      \quad |x| > \beta
    \end{cases}
  \end{align*}
  This map $T$ is inspired by the discretization of Lipschitz functions in~\cite{colt}.
  
  We claim that $Tf$ is $L$-Lipschitz for $L = 2(\ln\kappa + 1)$. It is clear that when $x,y \in [-\beta,\beta]$,
  \begin{align*}
    |Tf(x)-Tf(y)| \leq \frac{1}{\beta}|x-y| = (\ln\kappa + 1)|x-y|.
  \end{align*}
  Suppose that $1 \geq x > y \geq \beta$. Note that $\kappa^{\frac{x-\beta}{1-\beta}}
  = \kappa^x e^{x-1}$, thus
  \begin{align*}
    |Tf(x)-Tf(y)| &= \abs*{ \frac{f(\kappa^x e^{x-1} M)}{\kappa^x e^{x-1} M} -
    \frac{f(\kappa^y e^{y-1} M)}{\kappa^y e^{y-1} M}} \\
    &= \abs*{ \frac{f(\kappa^x e^{x-1} M)}{\kappa^{x-y} e^{x-y}\kappa^y e^{y-1} M} -
    \frac{f(\kappa^y e^{y-1} M)}{\kappa^y e^{y-1}M}} \\
    &= \abs*{
    \frac{f(\kappa^x e^{x-1} M)-f(\kappa^y e^{y-1}M)}{\kappa^{x-y}e^{x-y}\kappa^y e^{y-1}M}
    -
    \frac{(\kappa^{x-y} e^{x-y} - 1)f(\kappa^y e^{y-1} M)}{\kappa^{x-y}e^{x-y}\kappa^y e^{y-1}M}}
    \\
    &\leq
    \frac{\kappa^x e^{x-1}M - \kappa^y e^{y-1}M}{\kappa^{x-y}e^{x-y}\kappa^y e^{y-1}M} +
    \frac{\kappa^{x-y}e^{x-y}-1}{\kappa^{x-y}e^{x-y}} \\
    &= 2\cdot \frac{\kappa^{x-y}e^{x-y}-1}{\kappa^{x-y}e^{x-y}} \\
    &= 2(1 - (\kappa e)^{-(x-y)}) \\
    &\leq 2(x-y)\ln(\kappa e).
  \end{align*}
  A similar bound can be established on $[-1,-\beta]$. This establishes
  that $Tf$ is $L$-Lipschitz for $L = 2(\ln\kappa + 1)$.

  Next we show that $\norm{f} \leq M \norm{Tf}_{L^\infty([-1,1])}$. Indeed, fix an $i$ and suppose that 
  $|f(x_i)| = \norm{f}_{L^\infty([-M_i,M_i])}$.

  \begin{mycase}
    \case $|x_i| \leq M$. Recall that $\lambda_i \in (0,1)$, we have
    \begin{align*}
      \lambda_i|f(x_i)| \leq |f(x_i)| = M \abs*{ (Tf)\mleft(\frac{\beta}{M}x_i \mright)} \leq M \norm{Tf}_{L^\infty([-1,1])}.
    \end{align*}

    \case $x_i > M$. We have
    \[
        f(x_i) = x_i \cdot (Tf)\left( \beta\left(1+ \ln \frac{x_i}{M}\right) \right).
    \]
    Hence,
    \begin{equation}\label{eqn:aux_norm_lemma_bound}
      \lambda_i |f(x_i)| \leq \abs{\lambda_i x_i} \norm{Tf}_{L^\infty([-1,1])}
      \leq M \norm{Tf}_{L^\infty([-1,1])} .
    \end{equation}
    \case $x_i < -M$. The argument is similar to that of $x_i > M$ and \eqref{eqn:aux_norm_lemma_bound} also holds.
  \end{mycase}
  Combining these cases yields that $\norm{f}\leq M\norm{Tf}_{L^\infty([-1,1])}$. 

  Therefore, for every $f\in \lip_1$, there exists $g\in \lip_1$ such that 
  \[
    \norm{f} \leq ML\norm{g}_{L^\infty([-1,1])},
  \]
  which implies that
  \[
    \mathcal{N}(\text{Lip}_1, \norm{\cdot}, \eps)
    \leq \mathcal{N}\left(\text{Lip}_1, \norm{\cdot}_{L^\infty([-1,1])},
    \frac{\eps}{2M(\ln\kappa + 1)}\right). \qedhere
  \]
\end{proof}

The following entropy bound for Lipschitz functions is standard; see, e.g., Lemma 4.5.18 of~\cite{Talagrand2021}.
\begin{lemma}[Talagrand~\cite{Talagrand2021}]\label{lem:lipschitz_entropy_bound}
  It holds for $\eps > 0$ that
  \[
    \ln \mathcal{N}(\lip_1, \norm{\cdot}_{L^\infty([-1,1])}, \eps)
    \lesssim \frac{1}{\eps}.
  \]
\end{lemma}

Now we are ready to upper bound the Dudley's integral for $\lip_1$ in \eqref{eqn:main_dudley_bound}.

\begin{lemma}\label{lem:dudley's integral bound}
Suppose that $p>1$ is a constant. Under the assumptions of Lemma~\ref{lem:D_2_bound}, 
it holds that
  \[
    \int_0^\infty \sqrt{\ln \mathcal{N}(\pi_1(T), d_T^\phi,\eps)}d\eps 
    \lesssim 
    \left( \left(\frac{d}{n}\right)^{\frac{1}{p}}
           d^{\frac{1}{2}-\frac{1}{2\vee p}} 
           R^{\frac{1}{p}} \ln \frac{n}{d}
    \right)^{\frac{p}{2} \wedge 1}
  \]
  and
  \[
    \diam(\pi_1(T), d_T^\phi) \lesssim \left( \left(\frac{d}{n}\right)^{\frac{1}{p}}
           d^{\frac{1}{2}-\frac{1}{2\vee p}} 
           R^{\frac{1}{p}} \ln \frac{n}{d}
    \right)^{\frac{p}{2} \wedge 1}
  \]
\end{lemma}

\begin{proof}
  It follows from the definition of $d_T$ that
  \begin{align*}
    (d_T(f_1,f_2))^\phi &= \sup_{x\in \pi_2(T)} \norm{\Lambda(f_1(Ax)-f_2(Ax))}_{I,\infty}^{\phi} \\
    &= 
    \left(\max_{i\in I} \Lambda_{ii} \sup_{x\in \pi_2(T)} |(f_1(Ax) - f_2(Ax))_{i}| \right)^{\phi} \\
    &\leq 
    \norm{f_1 - f_2}^{\phi},
  \end{align*}
  where
  \[
        \norm{f} = \max_{i\in I} \Lambda_{ii}\norm{f}_{L^\infty([-M_i,M_i])},\quad M_i = \sup_{x\in \pi_2(T)} \abs{(Ax)_i}.
  \]
  We know from Lemma~\ref{lem: lewis weight properties} that 
  \[
    \max_{i\in I} \Lambda_{ii} M_i \lesssim \left(\frac{d}{n}\right)^{\frac{1}{p}} d^{\frac{1}{2} - \frac{1}{2\vee p}} R^{\frac{1}{p}} =: M
  \]
  and from the fact that $\Lambda_{ii} \geq (n/d)^{-1/p}$ that
  \[
    \kappa = \frac{\max_{i\in I} M_i}{M} \lesssim \max_i \frac{1}{\Lambda_{ii}} \leq \left(\frac{n}{d}\right)^{\frac{1}{p}}.
  \]
  We apply Theorem~\ref{thm:aux_L_infty_norm} to the cases $p>2$ and $p\leq 2$ separately below.

  \begin{mycase}
    \case $1 < p \leq 2$. Let $\Delta$ be such that $\Delta^{\frac{2}{p}} = 2M \ln\kappa$.
    \begin{align*}
      \int_0^\infty \sqrt{\ln \mathcal{N}(\pi_1(T),d_T^\phi,\eps)} d\eps &\leq \int_0^\infty \sqrt{\ln
        \mathcal{N}(\lip_1, \norm{\cdot}^{\phi},\eps)} d\eps \\
      &= \int_0^\infty \sqrt{\ln \mathcal{N}\mleft(\lip_1, \norm{\cdot}^{\frac{p}{2}}, \eps\mright)} d\eps \\
      &= \int_0^\infty \sqrt{\ln \mathcal{N}\mleft(\lip_1, \norm{\cdot}, \eps^{\frac{2}{p}} \mright)} d\eps \\
      &\leq \int_0^\infty \sqrt{\ln \mathcal{N}\mleft(\text{Lip}_1, \norm{\cdot}_{L^\infty([-1,1])},\frac{\eps^{\frac{2}{p}}}{2M \ln\kappa} \mright)} d\eps \\
      &\leq \int_0^\Delta \frac{\sqrt{2M \ln\kappa}}{\eps^{\frac{1}{p}}} d\eps \\
      &\asymp_p \sqrt{M \ln\kappa}\cdot \Delta^{1-\frac{1}{p}} \\
      &= \Delta^{\frac{1}{p}}\cdot \Delta^{1-\frac{1}{p}} \\
      &= \Delta \\
      &= (2M \ln\kappa)^{\frac{p}{2}}.
    \end{align*}

    \case $p > 2$. Let $\Delta = 2M\ln\kappa$.
    \begin{align*}
      \int_0^\infty \sqrt{\ln \mathcal{N}(\pi_1(T), d_T^\phi,\eps)} d\eps &\leq \int_0^\infty \sqrt{\ln
        \mathcal{N}\mleft(\lip_1, \norm{\cdot}, \eps\mright)} d\eps \\
      &\lesssim \int_0^\infty \sqrt{\ln\mathcal{N}\mleft(\lip_1, \norm{\cdot}_{L^\infty([-1,1])},\frac{\eps}{2M \ln\kappa} \mright)} d\eps \\
      &\lesssim \int_0^\Delta \frac{\sqrt{2M \ln\kappa}}{\eps^{\frac{1}{2}}}
      d\eps \\
      &\asymp \sqrt{2M \ln\kappa}\, \Delta^{\frac{1}{2}} \\
      &= \Delta^{\frac{1}{2}} \Delta^{\frac{1}{2}} \\
      &= \Delta \\
      &= 2M\ln\kappa
    \end{align*}
  \end{mycase}
  Combining both cases leads to
  \[
    \int_0^\infty \sqrt{\ln \mathcal{N}(\pi_1(T),d_T^\phi,\eps)} d\eps \lesssim (M\ln\kappa)^{\phi},
  \]
  which is exactly the first claimed result. 

  Moreover, we know from the proof above that
  \[
    \diam(\pi_1(T), d_T^\phi) \leq \left(\diam(\lip_1, \norm{\cdot}\right)^\phi
    \lesssim  
    \left(M\ln \kappa \cdot \diam(\lip_1, \norm{\cdot}_{L^\infty([-1,1])})\right)^\phi
    \lesssim
    (M\ln \kappa)^\phi.
  \]
  This is exactly the second claimed result.
\end{proof}

\begin{corollary}\label{lem:dudley_integral_lip1}
Suppose that $p>1$ is a constant. Under the assumptions of \Cref{lem:D_2_bound}, it holds that
  \[
    (\alpha F)^{\theta}\int_0^\infty \sqrt{\ln \mathcal{N}(\pi_1(T), d_T^\phi,\eps)}d\eps \lesssim \alpha\Gamma \ln\frac{n}{d}
  \]
  and
  \[
    (\alpha F)^{\theta} \diam(\pi_1(T), d_T^\phi) \lesssim  \alpha\Gamma \ln\frac{n}{d}.
  \]
  where $\Gamma$ is defined as in \eqref{eqn:Gamma}.
\end{corollary}

\section{Main Error Bound for Upper Bound Analysis}\label{sec:main_error_bound}
Based on the preceding sections, we now upper bound $\Psi_3$ as follows.

\begin{lemma} \label{lem:Psi_3_final}
Suppose that $p > 1$ is a constant. Under the assumptions of \ref{lem:D_2_bound}, it holds with probability at least $1-\delta$ that
    \[
        \Psi_3 \lesssim \Gamma \ln\frac{n}{\eps d} \left( \ln^\frac{5}{4} d + \sqrt{\ln\frac{1}{\delta}} \right),
    \]
    where $\Gamma$ is as defined in \eqref{eqn:Gamma}.
\end{lemma}

\begin{proof}
For ease of notations again we let $M \coloneqq \varphi
R^{\frac{1}{p}} = (\frac{d}{n})^{\frac{1}{p}}d^{\frac{1}{2}-\frac{1}{2 \vee
p}}R^{\frac{1}{p}}$ and $\kappa \coloneqq R^{\frac{1}{p}}/M \asymp \poly(\frac{n}{d})$.

Combining \Cref{lem:dudley_integral_T} and \Cref{lem:dudley_integral_lip1} with \Cref{lem:main_dudley_bound}, we obtain that
\begin{equation}\label{eqn:dudley_product_space}
  \int_0^\infty \sqrt{\ln \cN(T, D_2, \eps)} d\eps
  \lesssim \alpha\Gamma \left(\ln\frac{n}{d} + (\ln^\frac{5}{4} d) \sqrt{\ln \frac{n}{\eps d}}\right) \lesssim \alpha \Gamma (\ln^\frac{5}{4} d) \ln \frac{n}{\eps d}.
\end{equation}
and
\begin{equation}\label{eqn:diameter_product_space}
    \diam(T,D_2) \lesssim \alpha\Gamma \ln\frac{n}{d}. 
\end{equation}
Combining \eqref{eqn:Psi_3_moment_rewritten}, \eqref{eqn:Psi_3_moment_dudley} and \eqref{eqn:dudley_product_space}, \eqref{eqn:diameter_product_space}, we obtain that
\[
\left(\E_S \Psi_3^\ell\right)^{\frac{1}{\ell}} \lesssim \Gamma\ln\frac{n}{\eps d} \left(\ln^\frac{5}{4} d + \sqrt{\ell} \right).
\]
Taking $\ell = \log(1/\delta)$ and applying Markov's inequality leads to the claimed result.
\end{proof}

All the results so far give rise to the following main theorem.

\MainErrorBound*

With the preceding theorem established, we immediately obtain the following two corollaries.
\begin{corollary}
  \label{subspace embedding}
  Suppose $\alpha \gtrsim \frac{d^{1 \vee \frac{p}{2}}}{n}$. When
  conditioned on the event that $\norm{S \Lambda Ax}_p^p
  \lesssim$ $\norm{\Lambda Ax}_p^p$ for all $x \in
  \mathbb{R}^d$, it holds with probability at least $1-\delta$ that
  \begin{align*}
    & \sup_{(f,x)\in T} \abs[\big]{\, 
      \norm{S \Lambda(f(Ax)-\bar{f}(A\bar{x}))}_p^p -
    \norm{\Lambda(f(Ax)-\bar{f}(A\bar{x}))}_p^p \, } \\
    \lesssim\, & \frac{d^{\frac{1}{2}}}{(\alpha n)^{\frac{1}{2 \vee
    p}}} R\, \ln \frac{n}{\eps d} \left( \ln^\frac{5}{4} d + \sqrt{\ln \frac{1}{\delta}} \right).
  \end{align*}
\end{corollary}

\begin{proof}
  In the above theorem take $v = \bar{f}(A\bar{x})$ and $\eps$ to
  be a constant. Then $\norm{S\Lambda (\bar{f}(A \bar{x})-v)}_{p}^p = V = 0$. We have the following upper bound:
  \begin{align*}
    \norm{S\Lambda(f(Ax)-\bar{f}(A\bar{x}))}_{p}^p
    &\leq 2^{p-1} (\norm{S\Lambda(f(Ax)-\bar{f}(Ax))}_p^p + \norm{S\Lambda(\bar{f}(Ax)-\bar{f}(Ax))}_p^p)\\
    &\leq 2^{p-1} (2^p\norm{S\Lambda Ax}_p^p + \norm{S\Lambda(Ax-A\bar{x})}_{p}^p) \\
    &\leq 2^{p-1} (2^pC_1\norm{\Lambda Ax}_p^p + C_1\norm{\Lambda(Ax-A\bar{x})}_{p}^p) \\
    &\leq 2^{p-1} (2^pC_1\norm{\Lambda Ax}_p^p + C_1 2^p(\norm{\Lambda Ax}_p^p + \norm{\Lambda A\bar{x}}_{p}^p))\\
    &\leq 4^p C_1 R =: F
  \end{align*}
  and thus the result follows from Theorem~\ref{thm:main bound}.
\end{proof}

\begin{corollary}
  \label{bound on dist f and opt}
  Suppose that $\alpha \gtrsim \frac{d^{1 \vee \frac{p}{2}}}{n
  \eps}$ and $F \gtrsim \eps R$. When conditioned on the event that
  \[
    \norm{S\Lambda (\bar{f}(A \bar{x})-b)}_{p}^p \lesssim
    \norm{\Lambda (\bar{f}(A \bar{x})-b)}_{p}^p \quad
    \text{and} \quad \sup_{(f,x)\in T} \norm{S\Lambda(f(Ax)-\bar{f}(A\bar{x}))}_{p}^p \leq F,
  \]
  it holds with probability at least $1-\delta$ that
  \begin{align*}
    \sup_{(f,x)\in T} \abs*{ \big(
      \norm{S \Lambda(f(Ax)-b)}_p^p - \norm{\Lambda(f(Ax)-b)}_p^p \big) - \big( \lVert S
      \Lambda(\bar{f}(A\bar{x})-b)\rVert_p^p  - \norm{\Lambda(\bar{f}(A\bar{x})-b)}_p^p\big)} \\
    \lesssim \Gamma\cdot \ln \frac{n}{\eps d} \left( \ln^\frac{5}{4} d  + \sqrt{\log \frac{1}{\delta}} \right)
  \end{align*}
\end{corollary}

\begin{proof}
  In Theorem~\ref{thm:main bound} take $v = b$. We have $V =
  \norm{\Lambda (\bar{f}(A\bar{x})-b)}_p^p$ and the result
  follows, noticing that the last term in the error bound of
  Theorem~\ref{thm:main bound} is the dominating term.
\end{proof}

\section{$(1+\eps)$-Approximation} \label{sec:(1+eps)-approximation}

\begin{theorem}
  \label{eps close bound}
  Let $p\geq 1$ be a constant, $A \in \mathbb{R}^{n \times d}$, $f,\bar{f} \in \lip_1$, $\eps \in (0,1)$ be sufficiently small, and
  $\Lambda$ be an $n \times n$ diagonal matrix satisfying $\Lambda_{ii} > 0$ and $w_i(\Lambda A) \lesssim d/n$ for all $1 \leq i \leq n$. 
  Let $S$ be $n \times n$ random diagonal matrix in which the diagonal entries are i.i.d.\  $\alpha^{-\frac{1}{p}}Ber(\alpha)$ variables, where 
  $\alpha \gtrsim \frac{d^{\frac{p}{2} \vee 1}}{n \eps^{p \vee 2}} \poly\log\frac{n}{\eps}$.

  If $\hat{x}, \bar{x} \in \mathbb{R}^d$ and $\hat{f}, \bar{f} \in \lip_1$ satisfy
  \[
    (\hat{f},\hat{x}) = \argmin_{x \in \mathbb{R}^d,\, f \in \lip_1} \norm{S\Lambda (f(Ax)-b)}_p^p + \eps
    \norm{\Lambda Ax}_p^p,
  \]
  and
  \[
    \norm{\Lambda (\bar{f}(A\bar{x})-b)}_p^p - \norm{\Lambda (\hat{f}(A\hat{x})-b)}_p^p \lesssim \eps (
      \norm{\Lambda (\bar{f}(A\bar{x})-b)}_p^p + \eps \norm{\Lambda A \bar{x}}_p^p  ),
  \]
  then with probability at least $0.9$, 
  \begin{align*}
    \abs*{ \big( \lVert S
      \Lambda(\hat{f}(A\hat{x})-b)\rVert_p^p - \norm{\Lambda(\hat{f}(A\hat{x})-b)}_p^p \big) - \big(
      \norm{S \Lambda(\bar{f}(A\bar{x})-b)}_p^p  - \norm{\Lambda(\bar{f}(A\bar{x})-b)}_p^p\big)} \\ \leq
    \eps \big( \norm{\Lambda (\bar{f}(A\bar{x})-b)}_p^p +
    \norm{\Lambda A \bar{x}}_p^p  \big).
  \end{align*}
\end{theorem}

\begin{proof}
  We assume that $p>1$. The proof for the case of $p=1$ is deferred to \Cref{sec:p=1}.

  By the optimality of $(\hat{f},\hat{x})$, we have
  \begin{align*}
    \norm{S \Lambda(\hat{f}(A\hat{x})-b)}_p^p + \eps \norm{\Lambda A \hat{x}}_p^p \leq \lVert S
    \Lambda(\bar{f}(A\bar{x})-b)\rVert_p^p + \eps \lVert \Lambda
    A \bar{x} \rVert_p^p
  \end{align*}
  By Markov's inequality, with probability at least $0.99$, 
  \[
    \norm{S \Lambda(\bar{f}(A\bar{x})-b)}_p^p \leq
    100\norm{\Lambda(\bar{f}(A\bar{x})-b)}_p^p.
  \]
  We conditioned on this event in the remainder of the proof. This implies that
  \begin{align*}
    \norm{\Lambda A \hat{x}}_p^p &\leq \frac{1}{\eps}
    \norm{S \Lambda(\bar{f}(A\bar{x})-b)}_p^p + \norm{\Lambda A \bar{x}}_p^p \\
    &\leq \frac{100}{\eps} \norm{\Lambda(\bar{f}(A\bar{x})-b)}_p^p + \lVert \Lambda A
    \bar{x} \rVert_p^p \\
    &\coloneqq R_0.
  \end{align*}
  Throughout the rest of the proof, we assume that
  $\alpha \gtrsim \frac{d^{1 \vee \frac{p}{2}}}{n \eps^{p \vee  2}} \poly(\ln n)$ and $\delta\sim 1/\log\log\frac{1}{\eps}$ so that the error term in Corollary~\ref{bound on dist f and opt} can be bounded as
  \[
    \Gamma \ln\frac{n}{\eps d}\cdot \left(\poly(\ln d) + \sqrt{\ln\frac{1}{\delta}}\right) \lesssim \eps F^\theta R^\beta,
  \]
  where $\beta = \frac{1}{2} \wedge \frac{1}{p}$ and $\theta = (1 - \frac{1}{p}) \vee \frac{1}{2}$. It is obvious that $\theta + \beta = 1$.

  \paragraph{Bounding $F$ in Corollary~\ref{bound on dist f and opt}.}
  Let $T_{-1} = \{ (f,x) \in \lip_1 \times \mathbb{R}^d: \norm{\Lambda Ax}_p^p \leq R_0\}$, where $R_0 = \frac{100}{\eps}
  \norm{\Lambda(\bar{f}(A\bar{x})-b)}_p^p + \lVert \Lambda
  A \bar{x} \rVert_p^p$.

  \medskip

  By Corollary~\ref{subspace embedding} with our choice of
  $\alpha$ and $R = R_0$, it holds with probability at least $0.99$ that
  \begin{equation}
    \label{event}
    \sup_{(f,x)\in T_{-1}} \abs{ \norm{S \Lambda(f(A x)- \bar{f}(A\bar{x}))}_p^p -
    \norm{\Lambda(f(Ax)- \bar{f}(A\bar{x}))}_p^p } \leq C_1 \eps R_0
  \end{equation}
  where $C_1$ is a constant that depend only on $p$. Below we
  shall use constants $C_2,C_3,\ldots$ to denote constants that
  depends only on $p$.

  Conditioning on the event in~\eqref{event}, we have
  \begin{align*}
     &\norm{\Lambda(\hat{f}(A\hat{x})- \bar{f}(A\bar{x}))}_p^p \\
    \leq\,& \lVert  S\Lambda(\hat{f}(A\hat{x})-
    \bar{f}(A\bar{x})) \rVert_p^p + C_1 \eps R_0 \\
    \leq\,& 2^{p-1}(\norm{S\Lambda(\hat{f}(A\hat{x})- b)}_p^p + \norm{S \Lambda(\bar{f}(A\bar{x})-b)}_p^p
    ) + C_1 \eps R_0 \\
    \leq\,& 2^{p-1} (\norm{S\Lambda(\bar{f}(A\bar{x})- b)}_p^p + \eps \norm{\Lambda A \bar{x}}_p^p +
    \norm{S \Lambda(\bar{f}(A\bar{x})-b)}_p^p ) +
    C_1 \eps R_0 \\
    =\,& 2^{p-1}(2\norm{S\Lambda(\bar{f}(A\bar{x})- b)}_p^p
    + \eps \norm{\Lambda A \bar{x}}_p^p) + C_1 \eps R_0 \\
    \leq\,& 2^{p-1}(2\cdot 100\norm{\Lambda(\bar{f}(A\bar{x})- b)}_p^p + \eps \norm{\Lambda A \bar{x}}_p^p)
    + C_1 \eps\left(\frac{100}{\eps} \norm{\Lambda(\bar{f}(A\bar{x})-b)}_p^p + \norm{\Lambda A\bar{x}}_p^p\right) \\
    \leq\,& C_2( \norm{\Lambda(\bar{f}(A\bar{x})- b)}_p^p + \eps \norm{\Lambda A \bar{x}}_p^p)
  \end{align*}
  for some large constant $C_2$. Define
  \[
    F_0 \coloneqq C_2(\norm{\Lambda(\bar{f}(A\bar{x})- b)}_p^p + \eps \norm{\Lambda A\bar{x}}_p^p).
  \]

  \paragraph{Defining $T_i$ and $R_i$ in Corollary~\ref{bound on dist f and opt}.}
  Recall that $R_0 = \frac{100}{\eps} \norm{\Lambda(\bar{f}(A\bar{x})-b)}_p^p + \lVert \Lambda A
  \bar{x} \rVert_p^p$.

  \medskip

  We shall define $R_i$ based on $R_{i-1}$ ensuring that $R_i
  \leq R_0$ and that each $R_i$ has the form $X_i \norm{\Lambda(\bar{f}(A\bar{x})-b)}_p^p + Y_i \lVert \Lambda A
  \bar{x} \rVert_p^p$ for some $X_i,Y_i \geq 1$. Furthermore let,
  \[
    T_i = \left\{ (f,x) \in \lip_1\times \mathbb{R}^d: \norm{\Lambda Ax}_p^p \leq R_i \quad \text{and} \quad \norm{\Lambda(f(Ax)-\bar{f}(A\bar{x}))}_p^p \leq F_0 \right\}
  \]
  so that $T_i \subseteq T_0$.

  It is obvious that $R_i \leq R_0 \lesssim \frac{1}{\eps} F_0$.
  We also have
  \[
    \norm{S \Lambda(\bar{f}(A\bar{x})-b)}_p^p \lesssim \norm{\Lambda(\bar{f}(A\bar{x})-b)}_p^p,
  \]
  and by~\eqref{event} and the fact that $T_i \subseteq T_{-1}$ that
  \[
    \sup_{\substack{(f,x) \in T_i}} \norm{S\Lambda(f(Ax)-\bar{f}(A\bar{x}))}_p^p 
    \leq
    \sup_{\substack{(f,x)\in T_i}} \norm{\Lambda(f(Ax)-\bar{f}(A\bar{x}))}_p^p + C_1 \eps R_0 
    \lesssim F_0.
  \]
  Using our choice of $\alpha$, $R=R_i$ and $F = F_0$, we have by Corollary~\ref{bound on dist f and opt} that with probability at least $1-\delta$
  \begin{equation}\label{eqn:star-star}
    \sup_{\substack{(f,x)\in T_i}} \abs*{(
      \norm{S \Lambda(f(Ax)-b)}_p^p - \norm{\Lambda(f(Ax)-b)}_p^p) 
      - (\norm{S \Lambda(\bar{f}(A\bar{x})-b)}_p^p  - \norm{\Lambda(\bar{f}(A\bar{x})-b)}_p^p)} 
    \leq C_3 \eps R_i^\beta F_0^\theta
  \end{equation}
  for some constant $C_3$.
  
  \paragraph{Bootstrapping.} We would like to argue that the solution $(\hat{f},\hat{x}) \in T_i$. For $T_0$ we have
  \[
    \norm{\Lambda A \hat{x}}_p^p \leq R_0 \quad \text{and}
    \quad \norm{\Lambda(\hat{f}(A\hat{x})-\bar{f}(A\bar{x}))}_p^p \leq F_0
  \]
  Thus $(\hat{f},\hat{x}) \in T_0$.

  From now on suppose that $(\hat{f},\hat{x}) \in T_i$ and we will
  argue that $(\hat{f},\hat{x}) \in T_{i+1}$. We will continue to
  bound \eqref{eqn:star-star}. Suppose that $\frac{KY_i}{X_i} \geq
  \eps$ for some $K \geq 1$. Then we can upper bound $R_i^\beta
  F_0^\theta$ as follows.
  \begin{align*}
    R_i^\beta F_0^\theta &= \biggl( X_i \norm{\Lambda(\bar{f}(A\bar{x})-b)}_p^p + Y_i \lVert \Lambda A
    \bar{x} \rVert_p^p \biggl)^\beta  \cdot C_2^\theta \biggl( \norm{\Lambda(\bar{f}(A\bar{x})-b)}_p^p + \eps \lVert \Lambda
    A \bar{x} \rVert_p^p \biggl)^\theta \\
    &\leq C_2^\theta \biggl( X_i \norm{\Lambda(\bar{f}(A\bar{x})-b)}_p^p + Y_i \lVert \Lambda A
    \bar{x} \rVert_p^p \biggl)^\beta \biggl( \norm{\Lambda(\bar{f}(A\bar{x})-b)}_p^p + \frac{KY_i}{X_i}
    \norm{\Lambda A \bar{x}}_p^p \biggl)^\theta \\
    &\leq \biggl(\frac{C_2}{X_i} \biggl)^\theta \biggl( X_i
      \norm{\Lambda(\bar{f}(A\bar{x})-b)}_p^p + Y_i \norm{\Lambda A \bar{x}}_p^p \biggl) \qquad 
      (\text{since $\beta + \theta = 1$})
  \end{align*}
  Thus, by the optimality of $\hat{f},\hat{x}$,
  \begin{align*}
    \norm{\Lambda A\hat{x}}_p^p &\leq \frac{1}{\eps}
    \biggl( \norm{S \Lambda(\bar{f}(A\bar{x})-b)}_p^p -
    \norm{S \Lambda(\hat{f}(A\hat{x})-b)}_p^p \biggl) +
    \norm{\Lambda A \bar{x}}_p^p \\
    &\leq \frac{1}{\eps} \biggl( \norm{\Lambda(\bar{f}(A\bar{x})-b)}_p^p - \norm{\Lambda(\hat{f}(A\hat{x})-b)}_p^p + C_3\eps R_i^\beta
    F_0^\theta\biggl) + \norm{\Lambda A \bar{x}}_p^p \quad (\text{by \eqref{eqn:star-star}})
  \end{align*}
  By our assumptions, it holds that
  \[
    \norm{\Lambda(\bar{f}(A\bar{x})-b)}_p^p - \norm{\Lambda(\hat{f}(A\hat{x})-b)}_p^p 
    \lesssim 
    \eps ( \norm{\Lambda(\bar{f}(A\bar{x})-b)}_p^p + \eps \norm{\Lambda A \bar{x}}_p^p ) 
    \lesssim 
    \eps R_i^\beta F_0^\theta.
  \]
  It thus follows that
  \begin{equation}\label{eqn:triple-star}
  \begin{aligned}
    \norm{\Lambda A \hat{x}}_p^p &\leq C_4(R_i^\beta
    F_0^\theta) + \norm{\Lambda A \bar{x}}_p^p \\
    &\leq C_4 \left(\frac{C_2}{X_i}\right)^\theta 
        \left(X_i\norm{\Lambda(\bar{f}(A\bar{x})-b)}_p^p + Y_i \norm{\Lambda A \bar{x}}_p^p\right) 
        + \norm{\Lambda A \bar{x}}_p^p \\
    &\leq X_{i+1}\norm{\Lambda(\bar{f}(A\bar{x})-b)}_p^p + Y_{i+1}\norm{\Lambda A \bar{x}}_p^p,
  \end{aligned}
  \end{equation}
  where
  \[
    X_{i+1} = C_4C_2^\theta X_i^{1-\theta}, \quad Y_{i+1} = 1 + \frac{C_4C_2^\theta KY_i}{X_i^\theta}.
  \]
  Define $R_{i+1}$ to be the minimum of $R_0$ and the right-hand side of \eqref{eqn:triple-star}, i.e.
  \[
    R_{i+1} = R_0 \wedge \left( X_{i+1}\norm{\Lambda(\bar{f}(A\bar{x})-b)}_p^p + Y_{i+1}\norm{\Lambda A \bar{x}}_p^p \right).
  \]
  We immediately have $(\hat{f},\hat{x}) \in T_{i+1}$ which is
  needed to iterate the argument. 
  
  Let $X_0 = \frac{100}{\eps}$
  and $Y_0 = 1$. By induction one can show that
  \[
    X_i = C_5^{\frac{1-(1-\theta)^i}{\theta}}
    \biggl(\frac{100}{\eps} \biggl)^{(1-\theta)^i}
  \]
  for $C_5 = C_4C_2^\theta$. Then $C_6 \leq X_i \leq
  \frac{100}{\eps}$ for some constant $C_6$ for all $i \leq r$,
  thus $Y_{i+1} \leq 1 + C_7Y_i \leq C_8Y_i$ for some constants
  $C_7$ and $C_8$.

  When $r \sim_p \ln \ln\frac{1}{\eps}$,
  \[
    X_r \leq C_9 \quad \text{and} \quad Y_r \leq C_4(C_8)^{r-1} =
    \poly\mleft(\ln \frac{1}{\eps}\mright).
  \]
  We shall also verify that $\frac{KY_i}{X_i} \geq \eps$ for small $\eps$. Indeed,
  \[
    \frac{Y_i}{X_i} \geq \frac{1}{\frac{(100C_5^\theta)}{\eps}}
    \geq \frac{\eps}{K}
  \]
  for $K = 100C_5^\theta$.

  We iterate the above argument $r$ times. The total failure probability is at most $\delta r + 0.03 = 0.1$ since $\delta \sim 1/r$. 
  It then follows from
  \eqref{eqn:star-star} with $i = r - 1$ that
  \begin{align*}
    & \norm{\Lambda (\hat{f}(A\hat{x})-b)}_p^p - \norm{\Lambda (\bar{f}(A\bar{x})-b)}_p^p \\
    \leq\,& \lVert S\Lambda
    (\hat{f}(A\hat{x})-b)\rVert_p^p - \lVert S\Lambda
    (\bar{f}(A\bar{x})-b)\rVert_p^p + C_3 \eps R_{r-1}^\beta F_0^\theta \\
    \stackrel{(a)}{\leq}\,& \eps \norm{\Lambda A \bar{x}}_p^p + C_3 \eps R_{r-1}^\beta F_0^\theta\\
    \stackrel{(b)}{\lesssim}\,& \eps \norm{\Lambda A \bar{x}}_p^p + \eps
    \big( X_r \norm{\Lambda (\bar{f}(A\bar{x})-b)}_p^p + Y_r \norm{\Lambda A \bar{x}}_p^p \big) \\
    \lesssim\,& \eps \norm{\Lambda (\bar{f}(A\bar{x})-b)}_p^p + \left( \eps \poly\mleft(\ln \frac{1}{\eps}\mright) \right)\norm{\Lambda A \bar{x}}_p^p,
  \end{align*}
  where step (a) is due to the optimality of $\hat{f}$ and $\hat{x}$ and step (b) to \eqref{eqn:triple-star}.
  Rescaling $\eps \poly(\ln\frac{1}{\eps})$ to $\eps$ proves the claimed result, with additional $\poly\log\frac{1}{\eps}$ factors in the lower bound for $\alpha$.
\end{proof}

We are now ready to prove our main theorem of $(1+\eps)$-approximation, Theorem~\ref{thm:1+eps approx}.

\MainUpperBound*

\begin{proof}
  Let $(\hat{f},\hat{x}) = \argmin_{x \in \mathbb{R}^d, f \in
  \text{Lip}_1} \norm{S \Lambda(f(Ax)-b)}_p^p + \eps
  \norm{\Lambda Ax}_p^p$ and $\OPT = \norm{\Lambda (\bar{f}(A\bar{x})-b)}_p^p$. By the optimality of $(\bar{f},\bar{x})$, we have
  \[
    \norm{\Lambda (\bar{f}(A\bar{x})-b)}_p^p - \norm{\Lambda (\hat{f}(A\hat{x})-b)}_p^p \leq 0
  \]
  By Theorem~\ref{eps close bound}, with probability at least $0.9$, it holds that
  \begin{align*}
    \abs*{ ( \norm{S \Lambda(\hat{f}(A\hat{x})-b)}_p^p - \norm{\Lambda(\hat{f}(A\hat{x})-b)}_p^p) 
    - (\norm{S \Lambda(\bar{f}(A\bar{x})-b)}_p^p - \OPT)} \leq \eps (\OPT +
    \norm{\Lambda A \bar{x}}_p^p).
  \end{align*}
  This implies that
  \begin{align*}
    \norm{\Lambda(\hat{f}(A\hat{x})-b)}_p^p - \OPT
    &\leq \norm{S \Lambda(\hat{f}(A\hat{x})-b)}_p^p -
    \norm{S \Lambda(\bar{f}(A\bar{x})-b)}_p^p + \eps
    (\OPT + \norm{\Lambda A \bar{x}}_p^p) \\
    &\leq \eps \norm{\Lambda A \bar{x}}_p^p + \eps (\OPT + \norm{\Lambda A \bar{x}}_p^p)
    \quad (\text{by the optimality of } \hat{f},\hat{x}) \\
    &\leq \eps \OPT + 2\eps \norm{\Lambda A \bar{x}}_p^p.
  \end{align*}
  Therefore,
  \[
  \norm{\Lambda(\hat{f}(A\hat{x})-b)}_p^p \leq (1 + \eps) \OPT + 2\eps \norm{\Lambda A \bar{x}}_p^p.
  \]
  Rescaling $\eps$ completes the proof.
\end{proof}

\section{The case $p=1$}\label{sec:p=1}

In this section, we shall first prove an analogous version of \Cref{lem:Psi_3_final} for $p=1$. 
The issue is that in \Cref{lem:dudley's integral bound}, our entropy bound $\ln \cN(\pi_1(T), d_T^\phi, \eps)$ grows too rapidly and is thus insufficient to control Dudley's integral 
\[
\int_0^\infty \sqrt{\ln \cN(\pi_1(T), d_T^\phi, \eps)}d\eps,
\]
since the resulting upper bound integral diverges.
To address this issue, we instead invoke the more sophisticated control \eqref{eqn:sup_moment_bound2} of the supremum of a subgaussian process, which allows
us to truncate the Dudley's integral at a suitably chosen $\eps_0 > 0$.

Specifically, we obtain that
\begin{equation}\label{eqn:Psi_3_moment_dudley_p=1}
\E_\xi \left(\sup_{(f,x)\in T}\abs[\big]{ \inner{\xi_I}{Z_I(f,x)} - \inner{\xi_I}{Z_I(\bar{f},\bar{x})} }\right)^\ell \\
    \leq 
    C^\ell ( A + B + \sqrt\ell \diam(T, D_2))^\ell,
\end{equation}
where
\begin{align*}
    A &:= \left(\E_\xi\left(\sup_{\substack{(f_1,x_1),(f_2,x_2)\in T\\ D_2((f_1,x_1),(f_2,x_2))\leq C(\alpha F)^{1/2}\eps_0}}\abs[\big]{ \inner{\xi_I}{Z_I(f_1,x_1) - Z_I(f_2,x_2)} }\right)^\ell\right)^{\frac{1}{\ell}},\\
    B &:= \int_{C(\alpha F)^{1/2}\eps_0}^\infty \sqrt{\ln\cN(T, D_2, \eps)} d\eps.
\end{align*}
Here, $C$ is the hidden constant in \eqref{eqn:D_2_intermediate} and the truncation parameter $\eps_0$ is to be determined later.

By repeating the arguments in \Cref{lem:main_dudley_bound}, \Cref{lem:dudley's integral bound} and \Cref{lem:dudley_integral_lip1}, we can bound the integral term $B$ as
\[
    B \lesssim \alpha\Gamma\sqrt{\ln\frac{n}{\eps d}}\left(\ln^{\frac{5}{4}}d + \ln \frac{\sqrt{M\ln\kappa}}{\eps_0}\right),
\]
where $M = (d/n) R$, $\Gamma = (d/(\alpha n))^{1/2} F^{1/2} R^{1/2}$ and $\kappa = n/d$. Note that $\alpha \Gamma = (\alpha F)^{1/2}M^{1/2}$.

Next we upper bound $A$. Recall that
\[
D_2((f_1,x_1),(f_2,x_2)) = \norm{Z_I(f_1,x_1) - Z_I(f_2,x_2)}_2.
\]
Applying Cauchy-Schwarz inequality yields
\[
A \leq \left(\E_\xi\left(\sup_{\substack{(f_1,x_1),(f_2,x_2)\in T\\ D_2((f_1,x_1),(f_2,x_2))\leq C(\alpha F)^{1/2}\eps_0}}\norm{\xi_I}_2 \norm{Z_I(f_1,x_1) - Z_I(f_2,x_2)}_2 \right)^\ell\right)^{\frac{1}{\ell}} \leq C\sqrt{\abs{I}}(\alpha F)^{\frac{1}{2}}\eps_0.
\]

Combining \eqref{eqn:Psi_3_moment_dudley_p=1} with \eqref{eqn:Psi_3_moment_rewritten} and \eqref{eqn:Psi_3_moment_dudley}, we obtain that
\[
    \left(\E_S\Psi_3^\ell\right)^{\frac{1}{\ell}} \lesssim \frac{1}{\alpha}\left(\E_S |I|^{\frac{\ell}{2}}\right)^{\frac{1}{\ell}}(\alpha F)^{\frac{1}{2}}\eps_0 + \Gamma \sqrt{\ln\frac{n}{\eps d}}\left(\ln^{\frac{5}{4}}d + \ln \frac{\sqrt{M\ln\kappa}}{\eps_0} + \sqrt{\ell}\right).
\]
Choose $\eps_0$ such that 
\[
\sqrt{\alpha n} \eps_0 (\alpha F)^{1/2} = \alpha \Gamma, \quad\text{that is,}\quad \eps_0 = \frac{\sqrt{M}}{\sqrt{\alpha n}}.
\]
Since $\abs{I} \sim \Bin(|J|, \alpha)$ and $|J|\leq n$, \Cref{lem:binomial_moment} implies that
\[
    \E |I|^{\frac{\ell}{2}} \leq \left(\alpha |J| + \frac{\ell}{2}\right)^{\frac{\ell}{2}}\leq (\alpha n \ell)^{\frac{\ell}{2}}.
\]
Substituting this bound yields
\[
    \left(\E_S\Psi_3^\ell\right)^{\frac{1}{\ell}} \lesssim \Gamma\sqrt{\ell} + \Gamma\sqrt{\ln\frac{n}{\eps d}}\left(\ln^{\frac{5}{4}}d + \ln (\alpha n \ln \kappa) + \sqrt{\ell}\right)
    \lesssim
    \Gamma \ln\sqrt{\frac{n}{\eps d}}\left(\ln^{\frac{5}{4}}d + \ln n + \sqrt{\ell}\right).
\]
Finally, setting $\ell = \log(1/\delta)$ and applying Markov's inequality, we conclude that with probability at least $1-\delta$,
\[
    \Phi_3 \lesssim \Gamma \sqrt{\ln\frac{n}{\eps d}} \left( \ln^\frac{5}{4} d + \ln n + \sqrt{\ln \frac{1}{\delta}} \right).
\]
Thus, an analogous version of \Cref{lem:Psi_3_final} holds for $p=1$ as well (with an additional $\log n$ term). 

The same arguments establish corresponding versions of the remaining results in \Cref{sec:main_error_bound}, which includes \Cref{thm:main bound}. As a result, we obtain the $p=1$ version of \Cref{eps close bound}, with the modification that the condition on $\alpha$ requires an additional $\log n$ factor.

\section{Missing Proofs for the Lower Bound}\label{sec:lower_appendix}

This appendix supplies the auxiliary estimates and the deferred proof of Lemma~\ref{lem:code-reduction}.
The proof has three steps. First, we bound $\norm{Ax_i^\ast}_p^p$, thereby controlling the additive term $\eps\norm{Ax_i^\ast}_p^p$ in the approximation guarantee. Second, we show that any $\hat{x}$ satisfying that guarantee must have $i\in H(\hat{x})$ while $H(\hat{x})$ contains at most $O(\eps^{-p})$ indices. Finally, we compare the coordinates queried on $b^{(0)}$ with those queried on each $b^{(i)}$ and use the bound on $|H(\hat{x})|$ to obtain the query lower bound.

We begin with the norm estimate used to control the additive benchmark term $\eps\norm{Ax_i^\ast}_p^p$.

\begin{lemma}\label{lem:A-norm}
	If $F_i(x)\leq B^p$, then $\|Ax\|_p^p\le C_p(B^p + N\tau^p)$, where $C_p > 0$ is a constant depending only on $p$.
\end{lemma}
\begin{proof}

Note that
\begin{align*}
    \|Ax\|_p^p
    & =
    |(Ax)_{0,+}|^p + |(Ax)_{0,-}|^p + \sum_{j=1}^N \bigg(|(Ax)_{j,+}|^p + |(Ax)_{j,-}|^p\bigg).
\end{align*}
Write $x=(z,s)$ with $s=x_d$. For the pair indexed by $j\in[N]$, set $r_j=\langle u_j,z\rangle$ and $y=\tau s/B$; then $(Ax)_{j,+}=r_j-y$ and $(Ax)_{j,-}=-r_j-y$.
Namely, for each term in the summation, we have
\begin{align*}
    |(Ax)_{j,+}|^p + |(Ax)_{j,-}|^p
    & =
    |r_j-y|^p + |r_j+y|^p.
\end{align*}
We would like to exploit the condition $F_i(x)\leq B^p$.
Note that
\begin{align*}
    F_i(x)
    & =
    |((Ax)_{0,+})^+ - B|^p + |((Ax)_{0,-})^+|^p +  \\
    & + 
    \bigg(|((Ax)_{i,+})^+ - B|^p + |((Ax)_{i,-})^+|^p\bigg) \\
    & +
    \sum_{j=1, j\neq i}^N \bigg(|((Ax)_{j,+})^+|^p + |((Ax)_{j,-})^+|^p\bigg)
\end{align*}
and, for each term in the summation, we have
\begin{align*}
    \begin{cases}
        |(r_j-y)^+ - B|^p + ((-r_j-y)^+)^p & \text{if $j = i$}\\
        |(r_j-y)^+|^p + ((-r_j-y)^+)^p & \text{else.}
    \end{cases}
\end{align*}

	Hence, we first prove an auxiliary inequality to establish the connection between them. For $r,y\in\R$ and $\lambda\geq 0$, it holds that
	\[
	|r-y|^p + |r+y|^p \leq C_p\left( |(r-y)^+ - \lambda|^p + ((-r-y)^+)^p + \lambda^p + |y|^p \right).
	\]
	Indeed, the left-hand side is
	\[
	\bigg((r-y)^+\bigg)^p + \bigg((y-r)^+\bigg)^p + \bigg((-r-y)^+\bigg)^p + \bigg((y+r)^+\bigg)^p.
	\]
	Note that
	\begin{gather*}
		(r-y)^+ \leq |(r-y)^+ - \lambda| + \lambda, \\
		(y-r)^+ = (-r-y + 2y)^+ \leq (-r-y)^+ + 2|y|, \\
		(y+r)^+ = (r-y + 2y)^+ \leq (r-y)^+ + 2|y|,
	\end{gather*}
	the claimed result follows easily.

	We now apply this inequality to the paired rows of $A$.  Taking $\lambda=B$, for the planted pair $j=i$, gives
	\[
	|(Ax)_{i,+}|^p + |(Ax)_{i,-}|^p \leq C_p\left( \left|((Ax)_{i,+})^+-B\right|^p + (((Ax)_{i,-})^+)^p + B^p + \left(\frac{\tau |s|}{B}\right)^p \right).
	\]
	For every unplanted pair $j\neq i$, the corresponding labels are zero, so the same inequality with $\lambda=0$ gives
	\[
		|(Ax)_{j,+}|^p + |(Ax)_{j,-}|^p \leq C_p\left( (((Ax)_{j,+})^+)^p + (((Ax)_{j,-})^+)^p + \left(\frac{\tau |s|}{B}\right)^p \right),
		\qquad j\in[N]\setminus\{i\}.
	\]
	By a similar argument to the proof of the auxiliary inequality, the anchor pair satisfies
	\[
		|(Ax)_{0,+}|^p + |(Ax)_{0,-}|^p \leq C_p\left( \left|((Ax)_{0,+})^+-B\right|^p + (((Ax)_{0,-})^+)^p + B^p\right).
	\]
	Consequently, after summing over all row pairs, every prediction term on the right-hand side is absorbed by $F_i(x)$, and we obtain
	\[
	\norm{Ax}_p^p \leq C_p\left( F_i(x) + B^p + N\left(\frac{\tau |s|}{B}\right)^p \right).
	\]
	Considering the $(0,\pm)$-th rows of $F_i(x) = \norm{(Ax)^+ - b^{(i)}}_p^p$, we see that
	\[
	2|s|^p \leq C_p(|s^+ - B|^p + ((-s)^+)^p + B^p) \leq C_p(F_i(x) + B^p) \leq 2C_p B^p.
	\]
	Overall, we conclude that
	\[
	\norm{Ax}_p^p \leq C_p'(B^p + N\tau^p).
	\]
\end{proof}

Lemma~\ref{lem:A-norm} applies in particular to $x_i^\ast$ because \eqref{eqn:opt-upper} gives $F_i(x_i^\ast)=\OPT_i\leq B^p$. We use this observation next to show that any $\hat{x}$ satisfying the approximation guarantee \eqref{eqn:good-solution} produces a candidate set $H(\hat{x})$ that contains the planted index $i$ but has size $O_p(\eps^{-p})$.

\begin{lemma}\label{lem:good-solution-list}
	There exist constants $\alpha = \alpha(p), C_p, \kappa_p > 0$ such that, if $N\tau^p\leq \kappa_p\eps^{-p}$ and $\hat{x}$ satisfies
	\begin{align}\label{eqn:good-solution}
		F_i(\hat{x}) &\le	\OPT_i+\eps(\OPT_i+\norm{Ax_i^\ast}_p^p),
	\end{align}
	then
	\begin{equation*}
		i\in H(\hat{x}).
	\end{equation*}
    Recall that $F_i$ depends on $B = \frac{\alpha}{\eps}$.
    Also, $|H(\hat{x})|\le C_p\eps^{-p}.$
\end{lemma}
\begin{proof}[Proof of {Lemma~\ref{lem:good-solution-list}}]
	By \eqref{eqn:opt-upper}, $\OPT_i\leq B^p$ and so $F_i(x_i^\ast)\leq B^p$. Lemma~\ref{lem:A-norm} therefore gives $\norm{Ax_i^\ast}_p^p\leq C_p(B^p+N\tau^p)$. It follows from
	\eqref{eqn:good-solution} that
	\begin{equation}\label{eqn:good-loss-upper}
		\begin{aligned}
			F_i(\hat{x})
			&\leq (B-(1-\tau))^p + \eps\left( (B-(1-\tau))^p + C_p(B^p + N\tau^p) \right) \\
			&\leq (B-(1-\tau))^p + \eps(B^p + C_p(B^p + \kappa_p \eps^{-p})).
		\end{aligned}
	\end{equation}
	Thus, recall that $B = \frac{\alpha}{\eps}$, we have
	\[
	F_i(\hat{x})	- (B-(1-\tau))^p \leq C_p' \eps^{1-p} (\alpha^p + \kappa_p).
	\]
	We first show that the planted index belongs to the list. Suppose that $i\notin H(\hat{x})$. By the definition of $H(\hat{x})$, we then have
	\[
	F_i(\hat{x}) \geq \left(B-\frac{1-\tau}{2}\right)^p.
	\]
	Thus
	\begin{align*}
		&\quad\ F_i(\hat{x}) - \left(B-(1-\tau)\right)^p \\
		&\geq \left(B-\frac{1-\tau}{2}\right)^p - \left(B-(1-\tau)\right)^p \\
		&\geq p\left(B-(1-\tau)\right)^{p-1}\frac{1-\tau}{2} \\
		&\geq p\left(\frac{B}{2}\right)^{p-1}\frac{1-\tau}{2} & \text{since $B\geq 2$ and $\tau \leq \frac{1}{4}$}\\
		&\geq \frac{3p}{4\cdot 2^p}\alpha^{p-1}\eps^{1-p} & \text{recall that $B = \frac{\alpha}{\eps}$}.
	\end{align*}
	We get a contradiction by choosing $\alpha$ sufficiently small, and then choosing $\kappa_p$ sufficiently small relative to $\alpha^p$. Therefore, $i\in H(\hat{x})$.

	It remains to bound the size of the list. Consider $j\in H(\hat{x})\setminus\{i\}$. Since $b_{j,+}^{(i)} = 0$,
	\[
	|((A\hat{x})_{j,+})^+-b_{j,+}^{(i)}|^p \geq \left(\frac{1-\tau}{2}\right)^p \geq \left(\frac{3}{8}\right)^p.
	\]
	Hence
	\[
	(|H(\hat{x})|-1)\left(\frac{3}{8}\right)^p \leq F_i(\hat{x}).
	\]
	Combining this with \eqref{eqn:good-loss-upper} and $B=\alpha/\eps$ gives
	\[
	|H(\hat{x})| \leq C_p'' \eps^{-p},
	\]
	where $C_p''>0$ depends on $\alpha$ and $\kappa_p$.
\end{proof}

We can now prove Lemma~\ref{lem:code-reduction}. The key point is that $b^{(i)}$ and $b^{(0)}$ differ only at $(i,+)$. Thus, unless the algorithm queries $(i,+)$, it observes the same values and produces the same output on these two inputs. Lemma~\ref{lem:good-solution-list} limits the number of planted indices $i$ for which this common output can satisfy the approximation guarantee.
\begin{proof}[Proof of Lemma~\ref{lem:code-reduction}]
	Let $\alpha$, $C_p$, and $\kappa_p$ be the constants from Lemma~\ref{lem:good-solution-list}, and set $\eps_{0,p}:=\alpha/2$. For each $i\in[N]$, let $\hat{x}^{(i)}$ be the algorithm's output on $b^{(i)}$. Define the set of planted indices on which the algorithm succeeds by
	\[
	S:=\left\{i\in[N]:F_i(\hat{x}^{(i)})\leq
	\OPT_i+\eps(\OPT_i+\norm{Ax_i^\ast}_p^p)\right\}.
	\]
	Because $I$ is uniform on $[N]$, the assumed success probability gives
	\[
	|S|\geq \frac{4}{5}N.
	\]

	We now compare these successful executions with the baseline execution. Run the algorithm on $b^{(0)}$. Let $Q\subseteq[N]$ be the set of indices $j$ for which the algorithm reads the entry corresponding to $(j,+)$, and let $\hat{x}^{(0)}$ be its output. Then $|Q|\leq q$.

	For each $i\notin Q$, the algorithm sees the same sequence of queried coordinates and responses on $b^{(i)}$ and $b^{(0)}$, since these vectors differ only at the unqueried entry $(i,+)$. Thus $\hat{x}^{(i)}=\hat{x}^{(0)}$. In particular, if $i\in S\setminus Q$, Lemma~\ref{lem:good-solution-list} implies that $i\in H(\hat{x}^{(0)})$. If $S\setminus Q$ is nonempty, applying that lemma to any $i\in S\setminus Q$ also gives $|H(\hat{x}^{(0)})|\leq C_p\eps^{-p}$. Therefore,
	\[
	|S| = |S\cap Q| + |S\setminus Q| \leq |Q| + C_p\eps^{-p},
	\]
	where the inequality is immediate when $S\setminus Q$ is empty. Consequently,
	\[
	|Q| \geq \frac{4}{5} N - C_p\eps^{-p}.
	\]
	Thus all but at most $C_p\eps^{-p}$ successful indices must lie in $Q$. Since $N\ge K_p\eps^{-p}$, choosing $K_p$ sufficiently large gives $|Q| = \Omega_p(N)$, and hence $q=\Omega_p(N)$.
\end{proof}

\end{document}